\documentclass[10pt]{article}
\usepackage{amsmath,amsthm,latexsym,amssymb,amsfonts,epsfig, psfrag,color}
\usepackage{graphicx}
\usepackage{caption}
\usepackage{subcaption}
\usepackage{sidecap}
\usepackage{wrapfig}
\usepackage{url}

\makeatletter \@addtoreset{equation}{section} \makeatother

\newtheorem{Theorem}{Theorem}[section]

\newtheorem{Proposition}{Proposition}[section]
\newtheorem{Example}{Example}[section]

\def\be{\begin{equation}}
\def\ee{\end{equation}}
\def\ba{\begin{eqnarray}}
\def\ea{\end{eqnarray}}

\newcommand\nn{\nonumber}
\newcommand\q{\quad}
\def\Nl{{\mathchoice
{\setbox0=\hbox{$\displaystyle\rm N$}\hbox{\hbox to0pt
{\kern0.4\wd0\vrule height0.9\ht0\hss}\box0}}
{\setbox0=\hbox{$\textstyle\rm N$}\hbox{\hbox to0pt
{\kern0.4\wd0\vrule height0.9\ht0\hss}\box0}}
{\setbox0=\hbox{$\scriptstyle\rm N$}\hbox{\hbox to0pt
{\kern0.4\wd0\vrule height0.9\ht0\hss}\box0}}
{\setbox0=\hbox{$\scriptscriptstyle\rm N$}\hbox{\hbox to0pt
{\kern0.4\wd0\vrule height0.9\ht0\hss}\box0}}}}
\def\Zl{{\mathchoice
{\setbox0=\hbox{$\displaystyle\rm Z$}\hbox{\hbox to0pt
{\kern0.4\wd0\vrule height0.9\ht0\hss}\box0}}
{\setbox0=\hbox{$\textstyle\rm Z$}\hbox{\hbox to0pt
{\kern0.4\wd0\vrule height0.9\ht0\hss}\box0}}
{\setbox0=\hbox{$\scriptstyle\rm Z$}\hbox{\hbox to0pt
{\kern0.4\wd0\vrule height0.9\ht0\hss}\box0}}
{\setbox0=\hbox{$\scriptscriptstyle\rm Z$}\hbox{\hbox to0pt
{\kern0.4\wd0\vrule height0.9\ht0\hss}\box0}}}}
\def\Ql{{\mathchoice
{\setbox0=\hbox{$\displaystyle\rm Q$}\hbox{\hbox to0pt
{\kern0.4\wd0\vrule height0.9\ht0\hss}\box0}}
{\setbox0=\hbox{$\textstyle\rm Q$}\hbox{\hbox to0pt
{\kern0.4\wd0\vrule height0.9\ht0\hss}\box0}}
{\setbox0=\hbox{$\scriptstyle\rm Q$}\hbox{\hbox to0pt
{\kern0.4\wd0\vrule height0.9\ht0\hss}\box0}}
{\setbox0=\hbox{$\scriptscriptstyle\rm Q$}\hbox{\hbox to0pt
{\kern0.4\wd0\vrule height0.9\ht0\hss}\box0}}}}
\def\Rl{{\mathchoice
{\setbox0=\hbox{$\displaystyle\rm R$}\hbox{\hbox to0pt
{\kern0.4\wd0\vrule height0.9\ht0\hss}\box0}}
{\setbox0=\hbox{$\textstyle\rm R$}\hbox{\hbox to0pt
{\kern0.4\wd0\vrule height0.9\ht0\hss}\box0}}
{\setbox0=\hbox{$\scriptstyle\rm R$}\hbox{\hbox to0pt
{\kern0.4\wd0\vrule height0.9\ht0\hss}\box0}}
{\setbox0=\hbox{$\scriptscriptstyle\rm R$}\hbox{\hbox to0pt
{\kern0.4\wd0\vrule height0.9\ht0\hss}\box0}}}}
\def\Cl{{\mathchoice
{\setbox0=\hbox{$\displaystyle\rm C$}\hbox{\hbox to0pt
{\kern0.4\wd0\vrule height0.9\ht0\hss}\box0}}
{\setbox0=\hbox{$\textstyle\rm C$}\hbox{\hbox to0pt
{\kern0.4\wd0\vrule height0.9\ht0\hss}\box0}}
{\setbox0=\hbox{$\scriptstyle\rm C$}\hbox{\hbox to0pt
{\kern0.4\wd0\vrule height0.9\ht0\hss}\box0}}
{\setbox0=\hbox{$\scriptscriptstyle\rm C$}\hbox{\hbox to0pt
{\kern0.4\wd0\vrule height0.9\ht0\hss}\box0}}}}
\def\Hl{{\mathchoice
{\setbox0=\hbox{$\displaystyle\rm H$}\hbox{\hbox to0pt
{\kern0.4\wd0\vrule height0.9\ht0\hss}\box0}}
{\setbox0=\hbox{$\textstyle\rm H$}\hbox{\hbox to0pt
{\kern0.4\wd0\vrule height0.9\ht0\hss}\box0}}
{\setbox0=\hbox{$\scriptstyle\rm H$}\hbox{\hbox to0pt
{\kern0.4\wd0\vrule height0.9\ht0\hss}\box0}}
{\setbox0=\hbox{$\scriptscriptstyle\rm H$}\hbox{\hbox to0pt
{\kern0.4\wd0\vrule height0.9\ht0\hss}\box0}}}}
\def\Ol{{\mathchoice
{\setbox0=\hbox{$\displaystyle\rm O$}\hbox{\hbox to0pt
{\kern0.4\wd0\vrule height0.9\ht0\hss}\box0}}
{\setbox0=\hbox{$\textstyle\rm O$}\hbox{\hbox to0pt
{\kern0.4\wd0\vrule height0.9\ht0\hss}\box0}}
{\setbox0=\hbox{$\scriptstyle\rm O$}\hbox{\hbox to0pt
{\kern0.4\wd0\vrule height0.9\ht0\hss}\box0}}
{\setbox0=\hbox{$\scriptscriptstyle\rm O$}\hbox{\hbox to0pt
{\kern0.4\wd0\vrule height0.9\ht0\hss}\box0}}}}
\newcommand{\cc}{\mathcal C}

\newcommand{\ci}{\mathcal I}

\newcommand{\cm}{\mathcal M}

\newcommand{\cp}{\mathcal P}
\newcommand{\cq}{\mathcal Q}

\newcommand{\ct}{\mathcal T}

\newcommand{\fh}{\mathfrak{h}}

\def\nn{\nonumber}

\newcommand{\eqa}{\begin{eqnarray}}
\newcommand{\neqa}{\end{eqnarray}}

\newcommand{\p}{\partial}

\def\f{\frac}

\usepackage{bbm}

\def\q{{\quad}}

\definecolor{bianca}{rgb}{0,0.,0.8}

\begin{document}

{\renewcommand{\thefootnote}{\fnsymbol{footnote}}

\title{Canonical linearized Regge Calculus: counting lattice gravitons with Pachner moves}
\author{Philipp A H\"ohn\footnote{e-mail address: {\tt phoehn@perimeterinstitute.ca}} \\
\small   Perimeter Institute for Theoretical Physics,\\
 \small 31 Caroline Street North, Waterloo, Ontario, Canada N2L 2Y5
}

\date{}

}

\setcounter{footnote}{0}
\maketitle

\vspace{-.9cm}

\begin{abstract}

We afford a systematic and comprehensive account of the canonical dynamics of 4D Regge Calculus perturbatively expanded to linear order around a flat background. To this end, we consider the Pachner moves which generate the most basic and general simplicial evolution scheme. The linearized regime features a vertex displacement (`diffeomorphism') symmetry for which we derive an abelian constraint algebra. This permits to identify gauge invariant `lattice gravitons' as propagating curvature degrees of freedom. The Pachner moves admit a simple method to explicitly count the gauge and `graviton' degrees of freedom on an evolving triangulated hypersurface and we clarify the distinct role of each move in the dynamics. It is shown that the 1--4 move generates four `lapse and shift' variables and four conjugate vertex displacement generators; the 2--3 move generates a `graviton'; the 3--2 move removes one `graviton' and produces the only non-trivial equation of motion; and the 4--1 move removes four `lapse and shift' variables and trivializes the four conjugate symmetry generators. It is further shown that the Pachner moves preserve the vertex displacement generators. These results may provide new impetus for exploring `graviton dynamics' in discrete quantum gravity models.

\end{abstract}

\section{Introduction}

The canonical formulation of a physical theory usually offers convenient tools for extracting its dynamical content and, at the same time, gives a clear picture of the time evolution of relevant degrees of freedom. Specifically, in gravitational physics, a Hamiltonian formulation -- with initial value problem and `equal time surfaces' -- allows for an intuitive picture of the dynamics and simplifies the identification and counting of physical degrees of freedom. Suggestively, the seminal paper \cite{Arnowitt:1962hi} by Arnowitt, Deser and Misner, introducing their canonical formulation of general relativity, carries the unequivocal title `The dynamics of general relativity'. The beauty of this formulation of general relativity lies in the fact that it gives the latter the interpretation of describing the dynamics of spatial hypersurfaces in spacetimes.

In this spirit, we shall attempt to explore the dynamics of Regge Calculus \cite{Regge:1961px,Williams:1996jb,Regge:2000wu}, the most well-known simplicial discretization of general relativity. More specifically, by building up on the general canonical formulation of Regge Calculus developed in \cite{Dittrich:2011ke,Dittrich:2013jaa}, we shall systematically investigate, in canonical language, the dynamics of perturbative Regge Calculus to linear order in an expansion around flat background solutions \cite{Rocek:1982fr,Rocek:1982tj,barrett1987fundamental,barrett1988convergence,barrett1988convergence2,Dittrich:2009fb,Bahr:2009ku}. This linearized regime is governed by an expansion of the Regge action to quadratic order around a flat background. The motivation is the same: to give an intuitive picture of the linearized Regge dynamics and to identify and clearly distinguish propagating `lattice graviton' from lapse and shift type gauge degrees of freedom.

The presence of curvature in a Regge triangulation generically breaks the continuum diffeomorphism symmetry \cite{Bahr:2009ku,Hamber:1996pj,Morse:1991te}. More precisely, for flat Regge solutions there exists a continuous gauge symmetry corresponding to displacements, within the flat embedding space, of vertices in the bulk of a triangulation which leave the geometry flat. This vertex displacement symmetry can be interpreted as the incarnation of diffeomorphisms in triangulated spacetimes. It persists in the linearized theory because in this regime solutions to the equations of motion are additive such that (linearized) solutions corresponding to displacements in flat directions can be added to solutions with linearized curvature -- without changing boundary data. However, to higher order in the expansion around flat triangulations, the symmetry becomes broken \cite{Dittrich:2009fb}. This has crucial consequences: since the presence of the vertex displacement gauge symmetry is configuration dependent, so is the dynamical content of Regge Calculus. 

This severely complicates a detailed and explicit account of the dynamics in full Regge Calculus beyond general aspects \cite{Dittrich:2011ke}. It is therefore instructive to restrict the dynamics to a well-defined sub-regime which permits an explicit exploration. This is where the linearized sector of Regge Calculus assumes a special role: it is the only regime in which there is gauge symmetry and at the same time a non-trivial propagation of geometric degrees of freedom. And it permits to solve the equations of motion. Furthermore, this near-flat sector of the theory may be relevant for the continuum limit where the diffeomorphism symmetry of classical general relativity ought to be restored and geometries are locally flat \cite{barrett1987fundamental,barrett1988convergence,barrett1988convergence2}.

Since the gauge symmetries are generically broken for curved Regge triangulations, first class symmetry generators for full 4D Regge Calculus do not arise \cite{Dittrich:2011ke,Dittrich:2013jaa}. However, the linearized theory, as we shall see, does feature proper symmetry generating constraints (see also \cite{Dittrich:2009fb}). These, finally, will help us to shed light on the concept of propagating lattice `gravitons' and their dynamics in Regge Calculus -- although their relation to the gravitons in continuum general relativity remains to be clarified. 

In order to afford a systematic and comprehensive account of canonical linearized Regge Calculus, we shall consider the simplicial dynamics generated by Pachner moves \cite{pachner1,pachner2,Dittrich:2011ke,Dittrich:2013jaa,Dittrich:2013xwa}. The Pachner moves are elementary and ergodic moves which locally change the spatial triangulation and constitute the most basic and general simplicial evolution scheme. As they change the spatial triangulation, they also lead to a varying number of degrees of freedom in discrete time which requires the notion of evolving phase and Hilbert spaces \cite{Dittrich:2011ke,Dittrich:2013jaa,Hoehn:2014fka,Hoehn:2014wwa,Hoehn:2014aoa}. As such, these moves can also be nicely interpreted as defining canonical coarse graining/lattice shrinking or refinement/lattice growing operations and thereby be used to study embeddings of coarser into finer phase or Hilbert spaces \cite{Dittrich:2012jq,Dittrich:2013jaa,Dittrich:2013xwa,Hoehn:2014fka,Hoehn:2014wwa}. Specifically, for loop quantum gravity such considerations have led to the discovery of a new geometric vacuum for the theory \cite{Dittrich:2014wpa} and a proposal for constructing both its dynamics and continuum limit \cite{Dittrich:2014ala,DG2}.

For 4D linearized Regge Calculus, the Pachner moves offer a simple and systematic method to count and describe the `generation' and `annihilation' of `lapse and shift' gauge modes and `graviton' degrees of freedom on the evolving spatial triangulated hypersurface. In summary, in this manuscript we shall
\begin{itemize}
\item elucidate the origin of the vertex displacement gauge symmetry in linearized Regge Calculus,
\item derive the (first class) constraints generating this symmetry for arbitrary triangulated hypersurfaces,
\item show that these constraints are preserved by the linearized dynamics,
\item identify `gravitons' as (potentially) propagating curvature degrees of freedom that are invariant under the vertex displacement symmetry,
\item demonstrate how to count such `gravitons' via Pachner moves, and
\item study the distinct role of each of the Pachner moves in the dynamics.
\end{itemize}

These results may provide new impetus for the discussion of `graviton propagators' in spin foam models of quantum gravity \cite{Rovelli:2005yj,Alesci:2007tg} -- which yield the Regge action in semiclassical expansions of the transition amplitudes \cite{Perez:2012wv,Conrady:2008mk,Barrett:2009gg} --, and, more generally, for a better understanding of discretization changing dynamics in discrete models of (quantum) gravity \cite{Thiemann:1996ay,Thiemann:1996aw,Thiemann:2007zz,Alesci:2010gb,Bonzom:2011jv,Bonzom:2011hm,Dittrich:2011ke,Dittrich:2013xwa,Hoehn:2014fka,Hoehn:2014wwa,Ziprick:2014kla}.

The remainder is organized as follows. In order to make this manuscript as self-contained as possible, we begin by reviewing Regge Calculus in section \ref{sec_regcal}, the Pachner move evolution scheme in section \ref{sec_pachner}, the general canonical formulation of Regge Calculus \cite{Dittrich:2011ke} in section \ref{sec_can}, and an argument from \cite{Dittrich:2009fb}, elucidating the relation between the contracted Bianchi identities and the vertex displacement symmetry of linearized Regge Calculus, in section \ref{bianchi}. (The acquainted reader may skip these sections.) Subsequently, in section \ref{sec_deg}, we discuss degeneracies of the Hessian and the Lagrangian two-form on flat background solutions, resulting from vertex displacement symmetry. Section \ref{sec_lincanon} introduces the canonical variables and constraints of linearized Regge Calculus, while section \ref{obs} identifies `lattice gravitons' as gauge invariant curvature degrees of freedom. In section \ref{sec_gravcount}, a general counting of `gravitons' and gauge degrees of freedom is carried out using the Pachner moves. A procedure to disentangle the gauge from the `graviton' variables of the linearized theory is provided in section \ref{sec_tmatrix}, before we finally discuss the canonical Pachner move dynamics of 4D linearized Regge Calculus in detail in section \ref{sec_linpach}. Section \ref{sec_disc} closes with a discussion and an outlook. Technical details of this paper have been moved to appendices \ref{app_deg} and \ref{app_linpach}.

\section{Synopsis of Regge Calculus}\label{sec_regcal}

The principal idea underlying Regge Calculus \cite{Regge:1961px,Williams:1996jb,Regge:2000wu} is to replace a given smooth four--dimensional space--time $(\cm,\mathbf{g})$ with $C^2$ metric $\mathbf{g}$ by a piecewise--linear metric living on a triangulation $\ct$ which is comprised of {\it flat} 4--simplices $\sigma$. The metric on $\ct$ is piecewise--linear because it is flat on every simplex $\sigma$ and simplices are glued together in a piecewise--linear fashion. In fact, Regge Calculus is usually considered on a fixed triangulation $\mathcal{T}$. This discrete (triangulated) space--time is commonly viewed as an approximation to the continuum space--time, but may equally well be taken as a regularization thereof. 

While the geometry within continuum general relativity can be entirely encoded in the metric $\mathbf{g}$, 
the length variables associated to all edges $e$ in the triangulation $\ct$, $\{l^e\}_{e \in \mathcal{T}}$, completely specify the (piecewise--linear) geometry of the triangulation (assuming generalized triangle inequalities are satisfied). 
In other words, the edge lengths $l^e$ of the triangulation completely encode the piecewise--linear flat metric living on $\ct$ and are therefore the configuration variables of standard (length) Regge Calculus. 
For formulations using other geometric variables (e.g., areas and angles) see, for instance, \cite{Dittrich:2008va,Bahr:2009qd}. 

The challenge in formulating the Regge action is to translate the Ricci curvature term $R$ from the Einstein--Hilbert action into the triangulation. To this end, note that from the edge lengths one can compute any dihedral angle $\theta^\sigma_t$ around any triangle $t$ in $\ct$. $\theta^\sigma_t$ is the inner angle in the 4--simplex $\sigma$ between the two tetrahedra sharing the triangle $t$. Curvature in the discrete arises as follows: take a triangle $t$ in the bulk of $\ct$. $t$ will be contained in many 4--simplices. (Levi--Civita) parallel transporting a vector along a closed path around $t$ rotates the vector by the deficit angle
\ba\label{eps}
\epsilon_t=2\pi-\sum_{ \sigma \supset t} \theta^\sigma_t
\ea
in the plane perpendicular to $t$ (the sum ranges over all 4--simplices $\sigma$ containing $t$). $\epsilon_t$ measures the deviation from $2\pi$ of the sum of the dihedral angles around the triangle and thereby the intrinsic curvature concentrated at $t$.\footnote{This notion of curvature is distributional and only has support on the hinges because any closed path which is contractible (i.e., does not wind around a triangle) will yield a trivial holonomy.} That is, it is the bulk triangles which carry the curvature of the triangulation in the form of deficit angles in 4D Regge Calculus. 

The Regge action, without a cosmological constant term (and in Euclidean signature), defining a discrete spacetime dynamics for
a 4D triangulation $\ct$ with boundary $\partial \ct$ and interior
$\ct^\circ:=\ct\backslash \partial \ct$ is given by summing over all curvature contributions and consists of a bulk and a boundary term \cite{Regge:1961px,Hartle:1981cf}\footnote{We work in units of $c=8\pi G=1$.}
\begin{eqnarray}\label{regge1}
S_R\;=\;\sum_{t\subset
\ct^\circ}A_t\,\epsilon_t\;+\;\sum_{t\subset\partial \ct}A_t\,\psi_t\,,
\end{eqnarray}
where $A_t$ denotes the area of the triangle $t$ and appears because $\epsilon_t$ only has support on $t$. The boundary term is in shape identical to the bulk term, except that deficit angles are replaced by extrinsic angles, 
\ba\label{regge2}
\psi_t= \pi\;-\;\sum_{\sigma\subset
t}\theta_t^\sigma\qquad\text{for }t\subset\partial \ct
\ea
which measure the deviation from $\pi$ of the sum of dihedral angles around the triangles in the boundary of $\ct$. Notice that all quantities appearing in (\ref{regge1}) are functions of $\{l^e\}_{e \in \mathcal{T}}$.

Thanks to the boundary term in (\ref{regge1}) the action is {\it additive} in the following sense: if we glue two pieces of triangulations together then the total action contribution of the resulting glued triangulation is simply the sum of the action contributions of the two pieces of triangulation that we glued together.\footnote{This is not true for an action $\sum_t\sin(\epsilon_t)A_t$  although it converges to the Regge action in the limit of small deficit angles.} 
This will be important in the canonical Pachner move evolution below where we glue a single 4--simplex onto a 3D triangulated hypersurface $\Sigma$ during each move. The action of a 4--simplex $\sigma$ is a pure boundary term
\ba\label{regge3}
S_\sigma=   \sum_{t \subset \sigma} A_t \left(k_t \pi - \theta^\sigma_t \right)\,.
\ea
The coefficient $k_t$ depends on the gluing process: if $t$ is a new triangle in $\Sigma$ $k_t=1$. If, on the other hand, it is already present in $\Sigma$ before the gluing of $\sigma$, only the new dihedral angle of the simplex must be subtracted from the already present extrinsic angle (\ref{regge2}) so that in this case $k_t=0$.

Varying the action (\ref{regge1}) with respect to the lengths $l^e$ of the bulk edges $e\subset\ct^\circ$ and fixing the lengths of the boundary edges $e\subset\partial \ct$ yields the Regge equations of motion. To this end, the Schl\"afli identity,
\ba\label{schlaefli}
\sum_{t \subset \sigma} A_t\, \delta \theta^\sigma_t=0\,,
\ea
is essential which shows that the variation of the deficit angles vanishes. (It is the four--dimensional generalization of the two--dimensional fact that the sum of the dihedral angles in a triangle is constant.) 
This gives the Regge equation of motion the following form:
\begin{eqnarray}\label{regge4}
\sum_{t\supset e}\frac{\partial A_t}{\partial l^e}\,\epsilon_t=\;0  \,.
\end{eqnarray}

In this article we will be primarily concerned with perturbations (to linear order in an expansion parameter) around special Regge solutions, namely around flat solutions with vanishing deficit angles $\epsilon_t=0$, $\forall\,t\in\ct^\circ$. Abstractly, flat Regge triangulations only occur for special boundary configurations, but concretely, of course, any triangulation of flat 4D space will be such a solution. Flat solutions are not unique because vertices in the bulk can be displaced within the flat 4D embedding space without changing flatness. This is ultimately the reason for a vertex displacement gauge symmetry for flat configurations and linear perturbations around them. We shall discuss this in detail in the sequel. On the other hand, there is strong numerical evidence \cite{Bahr:2009ku} that curved solutions to (\ref{regge4}) are generally unique. In consequence, the gauge symmetries of the flat and linearized sector of Regge Calculus are broken to higher orders \cite{Dittrich:2009fb}.

We shall henceforth work in Euclidean signature to avoid subtleties arising from causal structure which would only unnecessarily cloud the main results. Note, however, that the canonical formulation below makes {\it a priori} no assumption about the signature and is equally applicable to Lorentzian signature.

\section{Canonical evolution from Pachner moves}\label{sec_pachner}

We shall review crucial ingredients for a canonical evolution scheme for simplicial gravity. A key hurdle to be overcome in a general formulation of canonical Regge Calculus is {\it the problem of foliations}: different 3D triangulated hypersurfaces of a generic foliation of a 4D Regge triangulation support different numbers of edges and thus configuration variables (see figure \ref{foliation}). A generic Regge triangulation implies varying numbers of degrees of freedom in `time' and requires the notion of {\it evolving phase spaces} \cite{Dittrich:2011ke,Dittrich:2013jaa,Hoehn:2014aoa}. Furthermore, hypersurfaces will generically overlap.
\begin{SCfigure}
\centering
\psfrag{t}{`time'}
{\includegraphics[scale=.3]{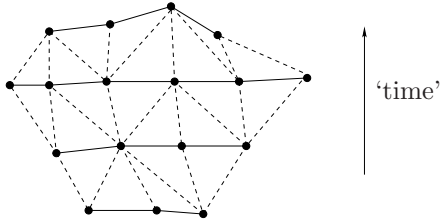}}
\hspace*{2cm}\caption{\small The {\it  problem of foliations}: in a generic triangulation, hypersurfaces overlap and are comprised of different numbers of simplices $\sigma$ which carry the variables of the theory. }\label{foliation}
\end{SCfigure}

First attempts \cite{Piran:1985ke,Friedman:1986uh} to circumvent this problem were based on the continuum 3+1 splitting and triangulating the spatial manifold $\Sigma$, but keeping a continuous time. Such an approach has two major problems: (1) The related discretization of the continuum first class constraint algebra leads to second class constraints (which are not automatically preserved under the dynamics) \cite{Piran:1985ke,Friedman:1986uh,Loll:1997iw}. (2) More importantly, the end result of such a continuous canonical evolution of a spatial triangulation clearly is not a four--dimensional spacetime triangulation and thus not compatible with the covariant picture. 

Instead, consistency requires that the canonical dynamics be equivalent to the covariant dynamics, directly following from the Regge action. In particular, a canonical dynamics consistent with the discrete action must produce a {\it discrete time evolution}. In contrast to the continuum, such a dynamics {\it cannot} be generated by a set of constraints via a Poisson bracket structure which necessarily has an infinitesimal action.\footnote{The exception are topological models for which a continuous time evolution can be recovered as a symmetry generated by constraints, namely the translation of vertices in time direction.} Rather, a well defined set of {\it evolution moves} is required to generate such a discrete dynamics; and any constraints arising from the Regge action should be consistent with the evolution moves \cite{Dittrich:2013jaa}.

What are possible evolution moves in simplicial gravity? We shall henceforth label discrete time steps by $k\in\mathbb{Z}$. Let $\Sigma_k$ be the 3D `spatial'\footnote{Since we work in Euclidean signature, we shall write `spatial' in quotation marks.} triangulated hypersurface of step $k$, constituting the `future boundary' of the triangulation `to the past' of $\Sigma_k$. An {\it evolution move} evolves $\Sigma_k$ to $\Sigma_{k+1}$ by gluing a 4D piece of triangulation $\ct_{k+1}$ to $\Sigma_k$ such that part of the boundary of $\ct_{k+1}$, consisting of tetrahedra, is identified with a subset of the tetrahedra of $\Sigma_k$. In particular, we disallow {\it singular evolution moves} which do not preserve the simplicial manifold property by identifying only lower than 3D subsimplices (triangles, edges or vertices) of the new simplex and the hypersurface. 

Obviously, there are many possibilities for such simplicial evolution moves. Fortunately, we can systematically handle all of them if we impose an additional restriction on the canonical dynamics: recall that in canonical general relativity one restricts from the outset to globally hyperbolic spacetimes with topology $\cm=\mathbb{R}\times\Sigma$, where $\Sigma$ is the spatial manifold. Similarly, we shall also restrict ourselves to evolution moves which preserve the `spatial' topology such that the 4D triangulation will likewise be of topology $\ct=\ci\times\Sigma$, where $\ci$ is some (closed) interval. 

Clearly, we could consider {\it global evolution moves} by fat slices which evolve an entire hypersurface at once such that $\Sigma_k\cap\Sigma_{k+1}=\emptyset$ (see figure \ref{cdtslices} for a schematic representation). While such global moves are relevant for the canonical dynamics \cite{Dittrich:2011ke,Dittrich:2013jaa}, as we shall see below, they are neither {\it elementary} because they can involve arbitrarily many 4--simplices, nor do they admit a nice geometric interpretation within the `spatial' hypersurfaces $\Sigma_k$. It is therefore difficult to give a systematic and completely general account of the canonical dynamics by means of global moves alone. 

\begin{SCfigure}
\psfrag{g}{glue}
\psfrag{sk}{$\Sigma_k$}
\psfrag{sk1}{$\Sigma_{k+1}$}
\includegraphics[scale=0.4]{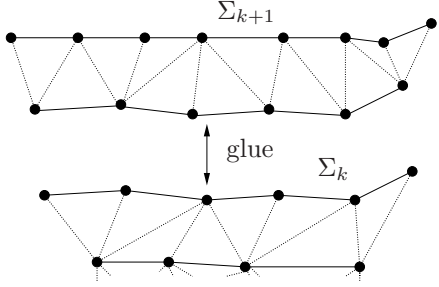}
 \hspace*{1.5cm} \caption{\small Global evolution by gluing fat slices at each step. This evolution is not elementary.}\label{cdtslices}
 \end{SCfigure}

Instead, a decomposition of the global moves into their smallest 4D building blocks, namely the 4-simplices, does allow for a clear and systematic description of the canonical dynamics under the above restrictions. We shall therefore consider the simplest evolution moves conceivable: gluing {\it locally} at each elementary time step $k$ a single 4-simplex $\sigma$ onto $\Sigma_k$ (see figure \ref{fig_glue}) \cite{Dittrich:2011ke,Hoehn:2011cm}. In this way, we evolve the hypersurface `forward in a multi-fingered time' through the full 4D Regge solution, in close analogy to canonical general relativity.
\begin{figure}[hbt!]
\begin{center}
\psfrag{g}{glue}
\psfrag{s}{$\sigma$}
\psfrag{s1}{$\Sigma_k$}
\psfrag{s2}{$\Sigma_{k+1}$}
\begin{subfigure}[b]{.22\textwidth}
\centering
\includegraphics[scale=.45]{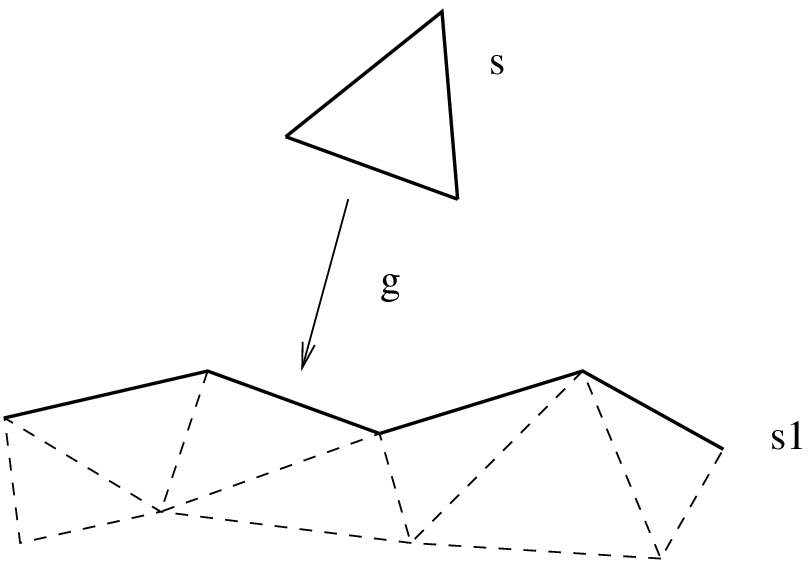}
\centering
\caption{\small }
\end{subfigure}
\hspace*{4.8cm}
\begin{subfigure}[b]{.22\textwidth}
\centering
\includegraphics[scale=.45]{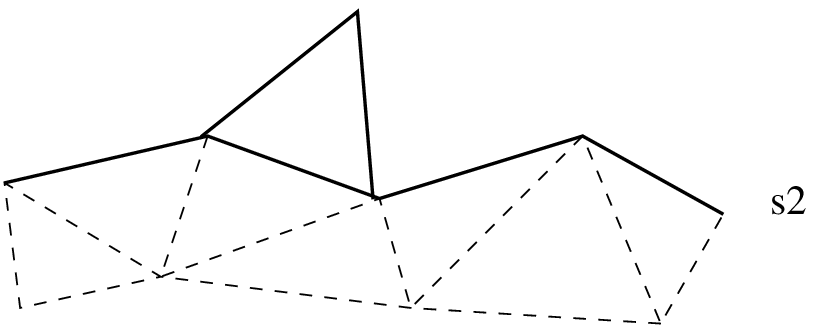}
\caption{\small }
\end{subfigure}
\caption{\small Schematic representation: at each time step $k$ a single 4--simplex $\sigma$ is glued to $\Sigma_k$ giving a new $\Sigma_{k+1}$.}\label{fig_glue}
\end{center}
\end{figure}
There are precisely four different possibilities for gluing a single 4--simplex $\sigma$ onto $\Sigma_k$: one can identify one, two, three or four of the five tetrahedra $\tau$ of $\sigma$ with tetrahedra in $\Sigma_k$, but clearly one cannot glue all five tetrahedra at once to $\Sigma_k$. These four different gluing types have the nice interpretation of 3D (or canonical) Pachner moves \cite{pachner1,pachner2} within $\Sigma_k$. More precisely, in order of increasing number of identified $\tau$ these are as follows (see \cite{Dittrich:2011ke} for further details).

\begin{description}
\item[1--4 Pachner move:] introduces 1 new vertex $v^*$ and 4 new boundary edges in $\Sigma_{k+1}$, but does not generate any internal triangles or bulk edges of $\ct$ (see figure \ref{14m}). Accordingly, this move neither produces curvature (deficit angles) nor a Regge equation of motion.
\item[2--3 Pachner move:] introduces 1 new boundary edge $n$ and 1 new bulk triangle $t^*$, but does not generate any internal edges (see figure \ref{23m}). Accordingly, this move produces curvature through a single new deficit angle, but does not give rise to a Regge equation of motion.
\item[3--2 Pachner move:] removes 1 old edge $o$ which becomes bulk of $\ct$ and introduces 3 new bulk triangles, but does not introduce new boundary edges (see figure \ref{23m}). Accordingly, this move produces one Regge equation of motion and three deficit angles.
\item[4--1 Pachner move:] removes 1 old vertex $v^*$ and 4 old edges and introduces 6 new bulk triangles, but does not introduce new boundary edges (see figure \ref{14m}). Accordingly, this move produces four Regge equations of motion and six deficit angles.
\end{description}
Notice that the 1--4 and 4--1 moves, as well as the 2--3 and 3--2 moves, are inverses of each other, respectively. All four Pachner moves are discretization changing evolution moves: they change the connectivity of $\Sigma$ and the number of edges contained in it and therefore produce a temporally varying discretization and number of degrees of freedom. In particular, the 1--4 and 2--3 moves introduce new but do not remove old edges and can thereby be interpreted as canonical refining (or lattice growing) moves. On the other hand, the 3--2 and 4--1 moves remove old but do not introduce new edges and may thus be viewed as canonical coarse graining (or lattice shrinking) moves \cite{Dittrich:2013xwa,Hoehn:2014fka,Hoehn:2014wwa,Hoehn:2014aoa}.
\begin{figure}[hbt!]
\begin{center}
\psfrag{v}{$v^*$}
\psfrag{1}{1--4}
\psfrag{2}{4--1}
\centering
\includegraphics[scale=.45]{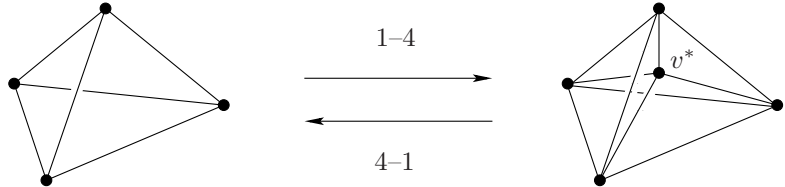}
\centering
\caption{\small The 1--4 and 4--1 Pachner moves within the 3D hypersurface $\Sigma$ are inverses of each other. }\label{14m}
\end{center}
\end{figure}

\begin{figure}[hbt!]
\begin{center}
\psfrag{t}{$t^*$}
\psfrag{n}{$n$}
\psfrag{1}{2--3}
\psfrag{2}{3--2}
\centering
\includegraphics[scale=.45]{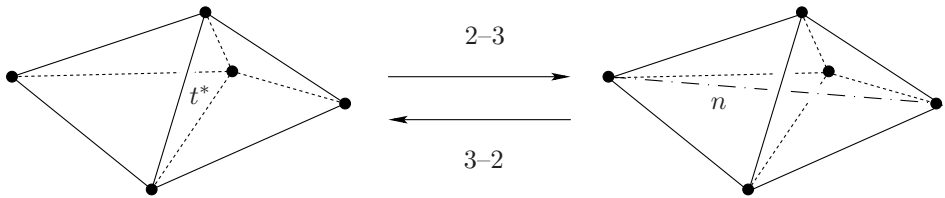}
\centering
\caption{\small The 2--3 and 3--2 Pachner moves within the 3D hypersurface $\Sigma$ are inverses of each other. For the 3--2 move replace the label $n$ by $o$ for the `old' removed edge.}\label{23m}
\end{center}
\end{figure}

The Pachner moves \cite{pachner1,pachner2} are exactly what is needed in order to construct the most general canonical evolution scheme for Regge Calculus. They are: 
\begin{itemize}
\item[(i)] {\it elementary} in that they involve only a single 4--simplex at the 4D level and between one and four tetrahedra within the 3D triangulated hypersurface $\Sigma$. That is, the moves are {\it local} in $\Sigma$.
\item[(ii)] {\it ergodic} piecewise-linear homeomorphisms; {\it any} two finite 3D triangulations $\Sigma$ and $\Sigma'$ of the same topology are connected via finite sequences of the Pachner moves. In particular, any other local or global evolution move can be decomposed into the latter.
\item[(iii)] applicable to arbitrary Regge triangulations.
\end{itemize}

Consequently, understanding the role of the four Pachner evolution moves in detail in canonical language means essentially understanding the full canonical dynamics of Regge Calculus. We shall review their canonical implementation and then focus on their specific role within linearized Regge Calculus.

\section{Review of canonical Regge Calculus}\label{sec_can}

In order to make this article self-contained, we shall summarize the canonical formulation of Regge Calculus, as developed in detail in \cite{Dittrich:2011ke,Dittrich:2013jaa,Hoehn:2011cm} and which we will later employ for the linearized theory.

Consider a {\it global evolution move} $0\rightarrow k$ from an initial triangulated hypersurface $\Sigma_0$ to another hypersurface $\Sigma_k$ as depicted in figure \ref{fig_global}. Denote by $S_{k}$ the piece of Regge action of the entire triangulation between $\Sigma_0$ and $\Sigma_k$. We shall label the edges in $\Sigma_k$ by $e$ and denote the corresponding lengths by $l^e_k$, while edges which are internal between $\Sigma_0$ and $\Sigma_k$ are labeled by $i$ and have lengths $l^i_k$. To distinguish the edges in the initial hypersurface $\Sigma_0$, we label these edges by $a$ and denote their lengths by $l^a_0$ such that $S_{k}=S_{k}(l^e_k,l^i_k,l^a_0)$. 
\begin{SCfigure}
\psfrag{sn0}{$\Sigma_0$}
\psfrag{sn1}{$\Sigma_k$}
\psfrag{s}{$S_{k}$}
\includegraphics[scale=.35]{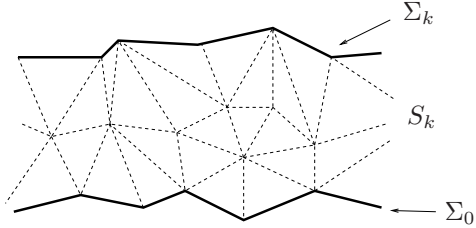} 
\hspace*{.5cm}\caption{\small Global evolution from $\Sigma_0$ to $\Sigma_k$. $S_k$ is the contribution of the Regge action to this piece of triangulation. For simplicity, we assume here $\Sigma_0\cap\Sigma_k=\emptyset$.}\label{fig_global} 
\end{SCfigure}
This piece of action can be taken as a Hamilton-Jacobi functional, i.e., generating function of the first kind, depending on `old' and `new' configuration variables. The canonical momenta conjugate to the lengths then read \cite{Dittrich:2011ke,Dittrich:2013jaa}
\ba
{}^+p^k_e=\f{\p S_k}{\p l^e_k},\q\q\q{}^+p^k_i=\f{\p S_k}{\p l^i_k}=0,\q\q\q{}^-p^0_a=-\f{\p S_k}{\p l^a_0}.\label{b6}
\ea
We emphasize that these equations are well defined regardless of the number of variables at the various time steps; they hold specifically for temporally varying `spatial' triangulations. The ${}^+p$ are called {\it post--momenta}, while the ${}^-p$ are called {\it pre--momenta}. In fact, the $(l^e_k,{}^+p^k_e)$ and $(l^i_k,{}^+p^k_i)$ form canonical Darboux coordinate pairs on the phase space $\cp_k:=T^*\cq_k$, where $\cq_k\simeq\mathbb{R}_+^{N_i+N_e}$ is the Regge configuration manifold of the edges labeled by $i$ and $e$ and $N_i,N_e$ are the numbers of edges labeled by $i,e$, respectively \cite{Dittrich:2013jaa}. Similarly, $(l^a_0,{}^-p^0_a)$ are canonical Darboux pairs on the phase space $\cp_0:=T^*\cq_0$, where $\cq_0\simeq\mathbb{R}_+^{N_a}$ is the Regge configuration manifold for $\Sigma_0$. In analogy to continuum general relativity, the pairs $(l^e_k,{}^+p^k_e)$ encode the intrinsic geometry of $\Sigma_k$ via $\{l^e_k\}_{e\subset\Sigma_k}$ and the extrinsic geometry of $\Sigma_k$ via $\{{}^+p^k_e\}_{e\subset\Sigma_k}$ --- an observation that becomes more apparent when writing $S_k$ in terms of the Regge action (\ref{regge1}) \cite{Dittrich:2011ke}. 

Notice that (\ref{b6}) defines an implicit Hamiltonian time evolution map $\fh_0:\cp_0\rightarrow\cp_k$; given initial data $(l^a_0,{}^-p^0_a)$ one can try to solve these equations for $(l^e_k,{}^+p^k_e)$ and $(l^i_k,{}^+p^k_i)$. This will clearly not give unique solutions in general. More precisely, if the coefficient matrices of the Lagrangian two-form on $\cq_0\times\cq_k$ (see especially appendix A of \cite{Dittrich:2013jaa}, but also \cite{Dittrich:2011ke,marsdenwest})
\ba
\Omega^k=-\frac{\partial^2\,S_{k}}{\partial l^e_k\partial l^a_0}\,dl^a_0\wedge dl^e_k-\frac{\partial^2\,S_{k}}{\partial l^{i}_k\partial l^a_0}\,dl^a_0\wedge d l^i_k\label{lagrange}
\ea
possess any left or right null vectors (degenerate directions) then the implicit function theorem implies that (\ref{b6}) does not define an isomorphism between $\cp_0$ and $\cp_k$. 

However, for a genuine canonical formulation of Regge Calculus we would rather like to have time evolution maps directly between the `equal time' hypersurfaces $\Sigma_0,\Sigma_k$. Noting that the second equation in (\ref{b6}) constitutes the Regge equations of motion for the bulk lengths $l^i_k$,\footnote{These lengths do not occur in any other action contribution.} one can solve these as a boundary value problem for $l^i_k(l^e_k,l^a_0,\kappa_k)$. This will generally not be uniquely possible in which case free parameters $\kappa_k$ must be chosen. Substituting this into the piece of Regge action $S_k$ yields an `effective' action (or Hamilton's principal function) $\tilde{S}_k(l^e_k,l^a_0)$ from which the $\kappa_k$ again drop out \cite{Dittrich:2013jaa}. This permits to define a new `effective' Hamiltonian time evolution map $\tilde{\fh}_0:\cp_0\rightarrow \tilde{\cp}_k$, where $\tilde{\cp}_k:=T^*\mathbb{R}_+^{N_e}$, via 
\ba
{}^+p^k_e=\f{\p\tilde{S}_k}{\p l^e_k},\q\q\q\q {}^-p^0_a=-\f{\p \tilde{S}_k}{\p l^a_0}\,,\label{b66}
\ea
i.e.\ for variables only associated to $\Sigma_0,\Sigma_k$. Note that $\dim\cp_0\neq\dim\tilde{\cp}_k$ is expressly allowed such that we have {\it evolving phase spaces}.

If the coefficient matrices of the `effective' Lagrangian two-form associated to $\tilde{S}_k$ \cite{Dittrich:2013jaa}
\ba
\tilde{\Omega}^k=-\frac{\partial^2\,\tilde{S}_{k}}{\partial l^e_k\partial l^a_0}\,d l^a_0\wedge d l^e_k\,,\label{lagrange2}
\ea
possess left or right null vectors, the implicit function theorem tells us that the left and right equations in (\ref{b66}) cannot each comprise an independent set such that there must exist primary constraint relations
\ba
{}^+C^k(l^e_k,{}^+p^k_e)=0,\q\q\q\q {}^-C^0(l^a_0,{}^-p^0_a)=0\label{pcons}
\ea 
among only variables of $\Sigma_k$ and $\Sigma_0$, respectively. These are the {\it post--} and {\it pre--constraints}, respectively, which form two first class sets of constraints \cite{Dittrich:2011ke,Dittrich:2013jaa}. In fact, the post--constraint surface $\cc_k^+\subset\tilde{\cp}_k$ is the image and the pre--constraint surface $\cc_0^-\subset\cp_0$ is the pre-image of $\tilde{\fh}_0$. It can be shown that $\tilde{\fh}_0$ preserves the symplectic structure restricted to these constraints. 

To every hypersurface $\Sigma_k$ there are associated {\it both} pre-- and post--momenta ${}^-p^k$ and ${}^+p^k$ and, if occurring, both pre-- and post--constraints ${}^-C^k,{}^+C^k$. For instance, if one continued the global evolution of figure \ref{fig_global} by another fat slice bounded by $\Sigma_k$ to the `past' and some $\Sigma_{k+X}$ to the `future', one likewise would have an `effective' Regge action contribution, call it $\tilde{S}_{k-}(l^e_k,l^{e'}_{k+X})$, associated to it. Equations (\ref{b66}) could then similarly be written for $\tilde{S}_{k-}$. It turns out that the Regge equations of motion for $e\subset\Sigma_k$, arising when gluing the new fat slice onto $\Sigma_k$, are equivalent to a {\it momentum matching} \cite{Dittrich:2011ke,Dittrich:2013jaa,Hoehn:2014aoa,marsdenwest}
\ba
p^k_e:={}^+p^k_e\overset{!}{=}{}^-p^k_e,\q\q\q\forall\,e\subset\Sigma_k\,.\label{mm}
\ea

As a consequence, on solutions to the equations of motion, both pre-- and post--constraints ${}^-C^k,{}^+C^k$ have to be implemented simultaneously. This has many non-trivial consequences, as the two sets may in general be independent. A thorough constraint analysis and classification has been devoted to this issue \cite{Dittrich:2013jaa,Hoehn:2014aoa} which we shall not repeat here. But in summary: (1) gauge symmetry generators are coinciding pre-- and post--constraints $C^k={}^-C^k={}^+C^k$ and first class (we shall see them as vertex displacement generators in linearized Regge Calculus below), (2) when taken together, ${}^-C^k,{}^+C^k$, can become second class, fixing free parameters. These two cases stand in close analogy to the continuum. However, there are two further possibilities which are specific to a temporally varying number of degrees of freedom in the discrete: (3) pre--constraints can be {\it coarse graining conditions} \cite{Dittrich:2013xwa,Hoehn:2014fka,Hoehn:2014wwa,Hoehn:2014aoa} which are first class (thus reducing the dynamical content by one canonical pair each) but no symmetry generators. These pre--constraints ensure that a finer 3D hypersuface $\Sigma_k$, carrying more dynamical data, can be consistently evolved into a coarser $\Sigma_{k+X}$ which can support less dynamical information than $\Sigma_k$. Similarly, (4) post--constraints can be {\it refining conditions} \cite{Dittrich:2013xwa,Hoehn:2014fka,Hoehn:2014wwa,Hoehn:2014aoa,Dittrich:2014ala,Dittrich:2014wpa,DG2} which are first class but not gauge symmetry generators. These guarantee that the smaller amount of dynamical information of a coarser $\Sigma_k$ can be consistently embedded in the larger phase space of a finer $\Sigma_{k+X}$.

The discretization changing dynamics in discrete gravity will generically generate such coarse graining and refining constraints and we shall also see them below in the Pachner moves. This is important because it entails that canonical constraints generically arise in Regge Calculus which are {\it not} the discrete analogue of Hamiltonian and diffeomorphism constraints, but which assume another crucial role. In particular, as mentioned in section \ref{sec_regcal}, the existence of the vertex displacement gauge symmetry (the discrete analogue of diffeomorphisms) is highly configuration dependent in Regge Calculus; this symmetry is broken for curved solutions \cite{Bahr:2009ku}. One can therefore not expect a discrete version of the Dirac hypersurface deformation algebra for generic Regge solutions \cite{Dittrich:2011ke,Bonzom:2013tna}. However, if the vertex displacement symmetry is present in the solution, the set of pre-- and post--constraints will contain the generators of these vertex displacements. In the sequel, we shall see this explicitly for linearized Regge Calculus for which these vertex displacement generators even satisfy an abelian algebra. This is as close as one comes to a consistent hypersurface deformation algebra in standard Regge Calculus with flat simplices.

In this article we will be concerned with the notion of propagating degrees of freedom on temporally varying discretizations. To this end, consider again the global evolution $0\rightarrow k$ of figure \ref{fig_global}. In the presence of pre-- and post--constraints (\ref{pcons}) there will generally not be sufficient equations in (\ref{b66}) for $\tilde{\fh}_0$ to be invertible. For every post--constraint ${}^+C^k$ there will exist an {\it a priori free} configuration datum $\lambda_k$ which cannot be {\it predicted} via $\tilde{\fh}_0$ \cite{Dittrich:2013jaa,Hoehn:2014aoa}. Similarly, for every pre--constraint ${}^-C^0$ there will exist an {\it a posteriori free} configuration datum $\mu_0$ that cannot be {\it postdicted} using $\tilde{\fh}_0$ and given `final' data $(l^e_k,{}^+p^k_e)$. {\it A priori} and {\it a posteriori free} data therefore do not correspond to degrees of freedom which propagate in the global move $0\rightarrow k$.  

Instead, the propagating data of the move $0\rightarrow k$ corresponds to those canonical data at $\Sigma_0$ and $\Sigma_k$ which can be uniquely mapped into each other, using the `effective' time evolution map $\tilde{\fh}_0$ (\ref{b66}). This set of propagating data is isomorphic to $\cp_0//\cc_0^-\simeq\tilde{\cp}_k//\cc_k^+$ and thus is incarnated at step $0$ as the set of Dirac {\it pre--observables} on $\cc_0^-$ which Poisson commute with all ${}^-C^0$ and at $k$ as the set of Dirac {\it post--observables} on $\cc^+_k$ which Poisson commute with all ${}^+C^k$ \cite{Dittrich:2013jaa,Hoehn:2014aoa}. In linearized Regge Calculus we shall see them as `lattice gravitons' below.

It is important to note that in Regge solutions with temporally varying `spatial' discretization $\Sigma$, the notion of a propagating degree of freedom as a pre-- or post--observable depends on the triangulated spacetime region and always requires {\it two} time steps -- in contrast to the continuum. For instance, the pre--observables at $\Sigma_0$ depend on the `future' hypersurface $\Sigma_k$. Gluing another fat slice onto $\Sigma_k$ which produces a new coarser `future' hypersurface $\Sigma_{k+X}$ and which thus comes with {\it coarse graining pre--constraints} at $k$, in fact, generates new `effective' coarse graining constraints at $\Sigma_0$ too in order to restrict the initial data on $\Sigma_0$ to the subset which consistently maps onto the coarse graining pre--constraints at $k$ \cite{Dittrich:2013jaa,Hoehn:2014fka,Hoehn:2014wwa,Hoehn:2014aoa}. That is, any pre--observable propagating from $0$ via $k$ to $k+X$ must also Poisson commute with the new `effective' pre--constraints on $\Sigma_0$. In this way, further evolution into the future can reduce the number of degrees of freedom which propagate from $\Sigma_0$ to the `future' hypersurface under consideration. Nevertheless, the dynamics is fully consistent -- albeit spacetime region dependent -- and reflects the fact that canonical coarse graining or refining operations change the dynamical content of a given 3D hypersurface $\Sigma_0$ (for further details, we refer the reader to \cite{Dittrich:2013jaa,Hoehn:2014fka,Hoehn:2014wwa,Hoehn:2014aoa,Dittrich:2013xwa}).

Finally, we review the implementation of the Pachner moves into canonical language. A Pachner move $\Sigma_k\rightarrow \Sigma_{k+1}$ can be viewed as locally updating the global move $\Sigma_0\rightarrow\Sigma_k$ to $\Sigma_0\rightarrow\Sigma_{k+1}$. Accordingly, the canonical data\footnote{We shall henceforth often drop the superscript ${}^+$ on the post--momenta, tacitly assuming that  (\ref{mm}) holds.} $(l^e_k,p^k_e)$ must be updated in the course of the move. To clarify this, we introduce a new notation specific to Pachner moves: edges in $\Sigma_k\cap\Sigma_{k+1}$ will be labeled by $e$, `new' edges introduced during the 1--4 or 2--3 moves which occur in $\Sigma_{k+1}$ but not in $\Sigma_k$ will be labeled by $n$, and `old' edges removed during the 2--3 and 4--1 moves which appear in $\Sigma_k$ but not in $\Sigma_{k+1}$ will be labeled by $o$. In order to describe the temporally varying number of variables, it is convenient to {\it extend} the phase spaces: for every pair $(l^n_{k+1},p^{k+1}_n)$ which appears at $k+1$ but not $k$ one can artificially introduce a new spurious pair $(l^n_{k},p^{k}_n)$ and thereby extend the phase space at $k$. Similarly, one can do the time reverse for canonical pairs labeled by $o$ such that one has {\it extended} phase spaces of equal dimension before and after the move.

The correct local time evolution equations are given by {\it momentum updating} which for the 1--4 and 2--3 Pachner moves (see figures \ref{14m} and \ref{23m}) is identical in shape and reads  \cite{Dittrich:2011ke,Dittrich:2013jaa}
\ba\label{4c6}
l^e_k&=&l^e_{k+1}\,,\q\q\, p^{k+1}_e\,=\,p^k_e+\frac{\partial S_\sigma(l^e_{k+1},l^n_{k+1})}{\partial l^e_{k+1}}\,,\\
p^{k}_n&=&0 \,,\q\q\q\,\,\, p^{k+1}_n\,=\, \frac{\partial S_\sigma(l^e_{k+1},l^n_{k+1})}{\partial l^n_{k+1}}\,,
\label{4c6c}
\ea
except that $n$ runs over four new edges for the 1--4 move and a single new edge in the 2--3 move. $S_\sigma(l^e_{k+1},l^n_{k+1})$ is the Regge action (\ref{regge3}) of the newly glued 4--simplex $\sigma$. (We refer to \cite{Dittrich:2011ke} for details on how to write (\ref{4c6}--\ref{c14c}) in terms of the Regge action.) The last equation in (\ref{4c6c}) defines four post--constraints ${}^+C^{k+1}_n:=p^{k+1}_n-\frac{\partial S_\sigma(l^e_{k+1},l^n_{k+1})}{\partial l^n_{k+1}}$ for the 1--4 and one similar post--constraint for the 2--3 move. In particular, there are no equations of motion involved in these moves such that $l^n_{k+1}$ is unpredictable and thus {\it a priori free}, given the data at $k$; notice that $l^n_{k+1}$ and ${}^+C^{k+1}_n$ are conjugate. (The spurious $l^n_k$ are pure gauge.) This has distinct consequences for the two types of moves:
\begin{description}
\item[1--4 move:] The four unpredictables $l^n_{k+1}$ are the lengths of the four new edges at the new vertex $v^*$. These can be interpreted as lapse and shift degrees of freedom. We shall see below that these are indeed gauge degrees of freedom in linearized Regge Calculus and the four corresponding ${}^+C^{k+1}_n$ turn into `discrete diffeomorphism' generators.

\item[2--3 move:] The deficit angle around the newly generated bulk triangle $t^*$ depends on the unpredictable $l^n_{k+1}$. Thus, the new curvature variable is itself {\it a priori free}. We shall see below that, in the linearized theory, this unpredictable deficit angle corresponds to a gauge invariant `graviton' and the corresponding single ${}^+C^{k+1}_n$ assumes the role of a {\it refinement consistency condition}.

\end{description}

On the other hand, for both 3--2 and 4--1 Pachner moves momentum updating is given by \cite{Dittrich:2011ke,Dittrich:2013jaa}
\ba\label{c14}
l^e_{k+1}&=&l^e_{k}\,,\q\q\q p^{k}_e\,=\,p^{k+1}_e-\frac{\partial S_\sigma(l^e_k,l^o_k)}{\partial l^e_{k}}\, ,\label{c14b}\\
p^{k+1}_o&=&0 \,,\q\q\q \,\,p^{k}_o\,=\,- \frac{\partial S_\sigma(l^e_k,l^o_k)}{\partial l^o_{k}}
  \, ,\label{c14c}
\ea
except that $o$ runs over a single edge for the 3--2 move and over four edges for the 4--1 move. The last equation in (\ref{c14c}) constitutes a single pre--constraint for the 3--2 and four pre--constraints for the 4--1 move. These pre--constraints ${}^-C^k_o:=p^k_o+\frac{\partial S_\sigma(l^e_k,l^o_k)}{\partial l^o_{k}}$ are equivalent to the Regge equations of motion for the new bulk edges labeled by $o$. (The spurious $l^o_{k+1}$ are pure gauge.) The repercussions are:
\begin{description}
\item[3--2 move:] The pre--constraint is generally non-trivial. It can be viewed as a {\it coarse graining consistency condition}. This will also be true in the linearized theory below where, furthermore, $l^o_k$ will correspond to an `annihilated lattice graviton'.

\item[4--1 move:] While the four ${}^-C^k_o$ may be non-trivial in full Regge Calculus, we shall see that they are automatically satisfied and `diffeomorphism' generators in the linearized theory. The $l^o_k$ are the lengths of the four edges adjacent to the removed vertex $v^*$, corresponding to lapse and shift.

\end{description}

The notion of propagation remains the same for Pachner moves: variables at $k+1$ which are predictable under momentum updating correspond to propagating degrees of freedom on the evolving phase spaces.

A final observation: the general pre-- and post--constraints (\ref{4c6c}, \ref{c14c}) of the Pachner moves are associated to {\it edges}. By contrast, the gauge symmetry of flat and linearized Regge Calculus corresponds to {\it vertex} displacements. We shall now turn to the linearized theory and explore how the vertex displacement generators can be produced from the edge pre-- and post--constraints.

\section{Bianchi identities and vertex displacement symmetry} \label{bianchi}

The contracted Bianchi identities, $\nabla^aG_{ab}=0$, of continuum general relativity follow from the diffeomorphism invariance of the Einstein--Hilbert action \cite{wald}. They constitute four differential relations among the ten Einstein field equations such that the latter are not fully independent and the ten metric components cannot be uniquely computed -- given enough boundary or initial data.

The situation is somewhat similar for Regge Calculus \cite{miller,Gentle:2008fy,Morse:1991te,Rocek:1982fr,Rocek:1982tj,Hamber:2001uq,Dittrich:2009fb}. In contrast to the continuum, however, the contracted Bianchi identities are satisfied as geometrical identities on rotation matrices rather than dynamical variables \cite{Hamber:2001uq,Freidel:2002dw} such that they generally do not render the Regge equations of motion interdependent. The exception is the regime of small deficit angles (and sufficiently `fat' simplices) where an approximate interdependence among the equations of motion arises thanks to the contracted Bianchi identities \cite{Morse:1991te}. In fact, this interdependence becomes exact in the linearized theory on a flat background triangulation. This is the origin of many of the special dynamical features of linearized Regge Calculus which shall be discussed in the core of this manuscript. For the sequel, it is therefore necessary to summarize the relation between the contracted Bianchi identities, the degeneracies of the Hessian of the flat background Regge action and the vertex displacement gauge symmetry of linearized Regge Calculus as originally clarified in \cite{Dittrich:2009fb}.

In a flat triangulation the bulk vertices can be freely displaced within the flat embedding space without affecting the (vanishing) deficit angles and thereby without violating the Regge equations of motion (\ref{regge4}). Such vertex displacements in flat directions are thus gauge symmetries of the Regge action {\it on flat solutions} \cite{Rocek:1982fr,Rocek:1982tj,Hamber:1992df,Dittrich:2008pw,Bahr:2009ku,Dittrich:2007wm}. Infinitesimally, for every vertex $v$, the corresponding length changes $\delta l^e$ of the adjacent edges are encoded in a set of four vectors $\delta l^e_{vI}=Y^e_{vI}$ ($I=1,\ldots ,4$), where
\be
Y^e_{vI}=\frac{\vec{B}_I\cdot \vec{E}_v^e}{\sqrt { \vec{E}_v^e \cdot  \vec{E}_v^e}}\,.
\ee
Every edge $e$ connected to $v$ is associated to a 4D vector $\vec{E}_v^e$ in the embedding space which points along the edge and whose length coincides with the edge length. The $\vec{B}_I$ comprise a basis in the 4D flat embedding space corresponding to four linearly independent displacement directions. Clearly, the components of $Y_{vI}^e$ for edges $e$ not connected to $v$ vanish. 

On flat backgrounds, the $Y^e_{vI}$ leave the deficit angles (\ref{eps}) around bulk triangles $t$ invariant
\ba\label{Yeps}
Y^e_{vI}\frac{\partial \epsilon_t}{\partial l^e}\,\,\Big|_\text{flat}=0\q \q \forall \q v,I,t\,.
\ea
Contracting with a factor $\partial A_t/\partial l^{e'}$ and summing over triangles, yields
\ba\label{bia3}
Y^e_{vI}  \sum_{t} \frac{\partial A_{t}}{\partial l^{e'}}  \frac{\partial \epsilon_t}{\partial l^e}\,\,\Big|_\text{flat}=0 \, .
\ea
Notice that $e'$ can also be a boundary edge.

On the other hand, thanks to the Schl\"afli identity (\ref{schlaefli}), the matrix of second partial derivatives of the Regge action (\ref{regge1}) on a flat triangulation, $\epsilon_t=0$, $\forall\,t$, reads
\ba\label{p2SR}
\frac{\partial^2 S_R}{\partial l^{e'}\partial l^e} =\sum_{t} \frac{\partial A_{t}}{\partial l^{e'}}  \frac{\partial \epsilon_t}{\partial l^e}\,\,\Big|_\text{flat}\,,
\ea
where $e,e'$ either are both bulk edges or one of them is a bulk and the other a boundary edge (if both edges were boundary edges this relation would not hold). Hence, (\ref{bia3}) shows that this matrix is degenerate with the vector fields $Y^e_{vI}$, $v\subset\ct^\circ$, defining the degenerate directions. Furthermore, (\ref{p2SR}) entails that the derivatives appearing in (\ref{bia3}) commute if both $e,e'$ are edges in the bulk $\ct^\circ$ of the triangulation. In particular, the Hessian of the Regge action is given by the matrix of second partial derivatives with respect to the bulk length variables, i.e.\
\ba
H_{ee'}:=\frac{\partial^2 S_R}{\partial l^{e'}\partial l^e} =\sum_{t} \frac{\partial A_{t}}{\partial l^{e'}}  \frac{\partial \epsilon_t}{\partial l^e}\,\,\Big|_\text{flat}\,\q\q e,e'\subset\ct^\circ\,,
\ea
such that (\ref{bia3}) implies that the Hessian is degenerate, $Y^e_{vI}H_{ee'}=0$, $v\subset\ct^\circ$. 

Next, in the linearized theory, the lengths $l^e={}^{(0)}l^e+ \varepsilon y^e+ O(\varepsilon^2)$ are expanded to linear order in an expansion parameter $\varepsilon$. The previous equations imply that the linearized Regge equations of motion are, in fact, {\it linearly dependent} with four relations per vertex $v$,
\ba\label{bia5}
Y^e_{vI}  \sum_{t} \frac{\partial A_{t}}{\partial l^{e}}  \frac{\partial \epsilon_t}{\partial l^{e'}}\,\,\Big|_\text{flat} \,\,y^{e'}=0 \, .
\ea
These equations constitute the linearized Bianchi identities and can be independently derived from a first order expansion of the `approximate Bianchi identities' \cite{miller,Gentle:2008fy,Morse:1991te}
\ba\label{001a}
Y^e_{vI}  \sum_{t} \frac{\partial A_{t}}{\partial l^{e}} \epsilon_t \approx 0 \, .
\ea
The linearized Bianchi identities (\ref{bia5}) and the degeneracy of the Hessian, $Y^e_{vI}H_{ee'}=0$, are
equivalent.

In conclusion, as a consequence of the vertex displacement gauge symmetry, in linearized Regge Calculus four length variables per vertex cannot be determined from the equations of motion. This symmetry is directly inherited from the flat background where, per bulk vertex, one obtains a four--parameter set of flat solutions so that the extremum of the action corresponding to flat solutions is not an isolated one, but admits four constant directions.

\section{Degeneracies of the Hessian and the Lagrangian two--form}\label{sec_deg}

In the sequel we shall frequently employ (effective) Hessian matrices and Lagrangian two--forms of flat Regge triangulations since these matrices of second partial derivatives of the Regge action evaluated on a flat background solution {\it define} the linearized dynamics. (In the linearized theory one expands the Regge action to quadratic order around a flat solution with the linear terms being zero on account of the background solution.) Consequently, we wish to examine the properties of these matrices for Regge Calculus evaluated on a flat solution and their relation with the vectors $Y_{vI}$ further.


To this end, consider again the global evolution $0\rightarrow k$ from section \ref{sec_can}. Pick a vertex $v\subset\Sigma_k$ and glue a piece $\ct_{\text{rest}}$ of flat triangulation onto $\Sigma_k$ such that this vertex $v$ becomes internal, i.e.\ such that the 4D star\footnote{The star of a vertex is the collection of all subsimplices in the triangulation which contain $v$ as a subsimplex.} of $v$ is completed and all edges adjacent to $v\subset\Sigma_k$ become bulk (see figure \ref{4dstarcomp}). We need the contribution from this additional piece of triangulation, in order to make use of the results of the previous section \ref{bianchi}. The particulars of $\ct_{\text{rest}}$ are not important; what matters is that $v$ becomes internal.
\begin{SCfigure}
\psfrag{sn0}{$\Sigma_0$}
\psfrag{sn1}{$\Sigma_k$}
\psfrag{t}{$\ct_{\text{rest}}$}
\psfrag{v}{$v$}
\psfrag{s}{$S_{k}$}
\includegraphics[scale=.35]{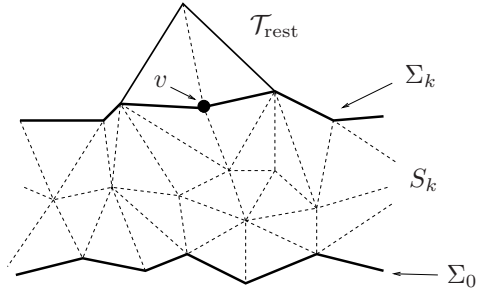} 
\hspace*{.5cm}\caption{\small Evolution from $\Sigma_0$ to $\Sigma_k$. Glue a piece of triangulation $\ct_{\text{rest}}$ onto $\Sigma_k$ in order to complete the 4D star of $v\subset\Sigma_k$.}\label{4dstarcomp} 
\end{SCfigure}

Using this and the results of section \ref{bianchi}, we show in appendix \ref{app_deg} that for a Hessian null vector $Y_{vI}$ at $v$, 
\ba
\Omega^{k}_{ae}Y^e_{vI}+\Omega^k_{ai}Y^i_{vI}=-Y^e_{vI}\frac{\partial^2\,S_{k}}{\partial l^e_k\partial l^a_0}-Y^{i'}_{vI}\frac{\partial^2\,S_{k}}{\partial l^{i'}_k\partial l^a_0}=0\,,
\ea
such that the relevant components of $Y_{vI}$ also define a {\it right null vector of the Lagrangian two--form} $\Omega^{k}$ (\ref{lagrange}) at step $k$. Furthermore, it is shown in appendix \ref{app_deg} that the `spatial' components of $Y^e_{vI}$ associated to the edges $e\subset\Sigma_k$ define also right null vectors of the `effective' Lagrangian two--form $\tilde{\Omega}^{k}$ (\ref{lagrange2}), 
\ba\label{omegadeg}
\tilde{\Omega}^{k}_{ae}Y^e_{vI}=-Y^e_{vI}\frac{\partial^2\,\tilde{S}_{k}}{\partial l^e_k\partial l^a_0}=0\,,
\ea
where (the $T_A$ are a maximal set of linearly independent non-degenerate directions of $\frac{\partial^2\,S_{k}}{\partial l^{i_1}_k\partial l^{i_2}_k}$ which are necessary in order to factor our the null direction of the latter matrix)
\ba
\tilde{\Omega}^k_{ae}&:=&-\frac{\partial^2\,\tilde{S}_{k}}{\partial l^e_k\partial l^a_0}=-\frac{\partial^2\,S_{k}}{\partial l^e_k\partial l^a_0}+\frac{\partial^2\,S_{k}}{\partial l^e_k\partial l^i_k}\,T^i_A\left(T^{i_1}_A\frac{\partial^2\,S_{k}}{\partial l^{i_1}_k\partial l^{i_2}_k}T^{i_2}_{A'}\right)^{-1}T^{i'}_{A'}\,\frac{\partial^2\,S_{k}}{\partial l^{i'}_k\partial l^a_0}\,.\label{eff2form}
\ea

In identical manner one shows that $Y_{vI}$ likewise defines a {\it left null vector of the Lagrangian two--form} $\Omega^{k+x}$ at step $k$ (and similarly of its `effective' version), where $\Omega^{k+x}$ arises from the action contribution $S_{k-}$ associated to the piece of `future' triangulation corresponding to the forward evolution from $\Sigma_k$ to some $\Sigma_{k+x}$ (see also the discussion in \cite{Dittrich:2011ke}).

In section \ref{bianchi} it was shown that the complete vectors $Y_{vI}$ (including the components associated to {\it all} bulk edges adjacent to $v$) are null vectors of the `bare' Hessian. In appendix \ref{app_deg} it is demonstrated that, similarly, the `spatial' components $Y^e_{vI}$ define degenerate directions of the `effective Hessian' $\tilde{H}_{ee'}$,
\ba\label{effYhesse}
Y^e_{vI}\tilde{H}_{ee'}=0\,, 
\ea
for which the lengths of all edges in the triangulation between $\Sigma_0$ and $\Sigma_k$ and in $\ct_{\text{rest}}$ are integrated out except the lengths $l^e_k$ of the `spatial' edges in $\Sigma_k$.

In the remainder of this article we shall work with the `effective' Hessian and Lagrangian two-forms of the Regge action. 
The remaining dynamical variables $l^e_k$ are associated to the hypersurface $\Sigma_k$ -- defining time step $k$ -- and thus the relevant variables for the canonical formulation. The matrices of second derivatives of the `effective' action will define the linearized canonical dynamics.

\section{Canonical variables and constraints of the linearized theory}\label{sec_lincanon}

The linearized theory of canonical Regge Calculus is given by an expansion of the canonical variables $l_k^e={}^{(0)}l_k^e+\varepsilon y_k^e+O(\varepsilon^2)$, $p^k_e={}^{(0)}p^k_e+\varepsilon\pi^k_e+O(\varepsilon^2)$, to linear order in some expansion parameter $\varepsilon$ around a flat background triangulation. ${}^{(0)}l^e_k,{}^{(0)}p^k_e$ denote the canonical variables of the flat background solution.

Consider again the evolution $0\rightarrow k$ from an initial triangulated hypersurface $\Sigma_0$ to another hypersurface $\Sigma_k$. Using this expansion, it follows from the expressions in (\ref{b6}) that the equations for the linearized pre-- and post--momenta now read\footnote{If $\Sigma_0\cap\Sigma_k\neq\emptyset$, we count any edges contained in this overlap simply to step $k$. That is, the corresponding linearized variables are among the $y^e_k,\pi^k_e$.}
\ba
{}^+\pi^k_e&=&\,\,\,\,\f{\p^2S_{k}}{\p l^e_k\p l^a_0}y^a_0+\f{\p^2 S_{k}}{\p l^e_k\p l^i_k}y^i_k +\f{\p^2S_{k}}{\p l^e_k\p l^{e'}_k}y^{e'}_k\,,\nn\\
{}^+\pi^k_i&=&\,\,\,\,\f{\p^2S_{k}}{\p l^i_k\p l^a_0}y^a_0+\f{\p^2 S_{k}}{\p l^i_k\p l^{i'}_k}y^{i'}_k +\f{\p^2S_{k}}{\p l^i_k\p l^{e'}_k}y^{e'}_k=0\,,\nn\\  
{}^-\pi^0_a&=&-\f{\p^2S_{k}}{\p l^a_0\p l^{a'}_0}y^{a'}_0-\f{\p^2 S_{k}}{\p l^a_0\p l^i_k}y^i_k -\f{\p^2S_{k}}{\p l^a_0\p l^{e'}_k}y^{e'}_k\,.
\ea
One can check that solving the equations of motion, ${}^+\pi^k_i=0$, for the bulk linearizations $y^i_k$ and inserting the solutions into the other two equations yields
\ba
{}^+{\pi}^k_e&=&\,\,\,\,\,\,N^k_{ee'}y^{e'}_k-\Omega^k_{ae}y^{a}_0\,,\nn\\
{}^-\pi^0_a&=&-M^0_{aa'}y^{a'}_0+\Omega^k_{ae}y^e_k\,.\label{reggelinmoms}
\ea
$\Omega^k_{ae}$ denotes here and henceforth the `effective' Lagrangian two-form (\ref{eff2form}) (for notational simplicity we drop the tilde from now on) and the matrices $N^k_{ee'},M^0_{aa'}$ are given by
\ba\label{matrixdefs}
N^k_{ee'}&:=&\,\,\,\frac{\partial^2\,\tilde{S}_{k}}{\partial l^e_k\partial l^{e'}_k}\,=\,\,\,\,\frac{\partial^2\,S_{k}}{\partial l^e_k\partial l^{e'}_k}-\frac{\partial^2\,S_{k}}{\partial l^e_k\partial l^i_k}\,T^i_A\left(T^{i_1}_A\frac{\partial^2\,S_{k}}{\partial l^{i_1}_k\partial l^{i_2}_k}T^{i_2}_{A'}\right)^{-1}T^{i'}_{A'}\,\frac{\partial^2\,S_{k}}{\partial l^{i'}_k\partial l^{e'}_k}\,,\nn\\
M^0_{aa'}&:=&\,\,\,\frac{\partial^2\,\tilde{S}_{k}}{\partial l^a_0\partial l^{a'}_0}\,=\,\,\,\,\frac{\partial^2\,S_{k}}{\partial l^a_0\partial l^{a'}_0}-\frac{\partial^2\,S_{k}}{\partial l^a_0\partial l^i_k}\,T^i_A\left(T^{i_1}_A\frac{\partial^2\,S_{k}}{\partial l^{i_1}_k\partial l^{i_2}_k}T^{i_2}_{A'}\right)^{-1}T^{i'}_{A'}\,\frac{\partial^2\,S_{k}}{\partial l^{i'}_k\partial l^{a'}_a}\,,
\ea
where, again, the $T_A$ define a maximal linearly independent set of non-degenerate directions of $\frac{\partial^2\,S_{k}}{\partial l^{i_1}_k\partial l^{i_2}_k}$ (see also appendix \ref{app_deg}). We emphasize that $N^k_{ee'}$ is {\it not} the `effective' Hessian $\tilde{H}^k_{ee'}$ (\ref{eff-hesse}) of the action at step $k$, such that generally $Y^e_{vI}N^k_{ee'}\neq0$. Instead, one would have $\tilde{H}^k_{ee'}=N^k_{ee'}+M^k_{ee'}$.

Contracting (\ref{reggelinmoms}) with the right and left null vectors $R_k,L_0$ of $\Omega^k$ yields the post-- and pre--constraints
\ba
{}^+C^k&=&(R_k)^e\left({}^+{\pi}^k_e-N^k_{ee'}y^{e'}_k\right)\,,\nn\\
{}^-C^0&=&(L_0)^a\left({}^-{\pi}^0_a+M^0_{aa'}y^{a'}_0\right)\,,\label{reggecon}
\ea
which must vanish and which only depend on the canonical variables from one time step. In particular, in section \ref{sec_deg} we saw that the `spatial' components $Y^e_{vI}$ define (i) right null vectors of the `effective' Lagrangian two-form ${\Omega}^k$ at $k$, (ii) left null vectors of the `effective' Lagrangian two-form ${\Omega}^{k+x}$ at $k$, and (iii) null vectors of the `effective' Hessian $\tilde{H}_{ee'}$ at $k$. This is important because it was shown in \cite{Hoehn:2014aoa} (see also the discussion in \cite{Dittrich:2013jaa}) that, in this case, the corresponding constraints
\ba
C^k_{vI}&=&(Y_k)^e_{vI}\left({}^+{\pi}^k_e-N^k_{ee'}y^{e'}_k\right)\,,\nn\\
C^0_{vI}&=&(Y_0)^a_{vI}\left({}^-{\pi}^0_a+M^0_{aa'}y^{a'}_0\right)\,,\label{CvI}
\ea
are, in fact, (1) simultaneously pre-- and post--constraints (accordingly we drop the ${}^\pm$ indices), (2) genuine gauge symmetry generators, (3) abelian first class,\footnote{$\{C^k_{vI},C^k_{v'I'}\}=0$ follows directly from $N^k_{ee'}=N^k_{e'e}$.} and (4) associated to genuine gauge variables $x^{vI}_k$.\footnote{In the classification of \cite{Hoehn:2014aoa}, these vectors and the corresponding symmetry generating constraints are of type (1)(A).} In section \ref{sec_linpach}, we shall discuss these gauge degrees of freedom further and show that these constraints are preserved by the linearized dynamics.

The two sets of constraints (\ref{CvI}), indeed, generate the displacement of the vertices in $\Sigma_0$ and $\Sigma_k$ in flat directions in the flat 4D embedding space: they lead precisely to the corresponding infinitesimal lengths changes of the edges $e\subset\Sigma_k$ or $a\subset\Sigma_0$ adjacent to the given vertex, 
\ba
\delta l^e_k=\{y^e_k,C^k_{vI}\}=(Y_k)^e_{vI}\,,\q\q\q\q \delta l^a_0=\{y^a_0,C^0_{vI}\}=(Y_0)^a_{vI}\,.
\ea
The contraction with the vectors $Y_{vI}$ associates these constraints invariably to {\it vertices} rather than {\it edges}.

For later purposes, let us now count how many linearly independent gauge generators $C^k_{vI}$ we have at step $k$. As seen in section \ref{bianchi}, there are exactly four vectors $Y_{vI}$ associated to each vertex $v$ in $\Sigma_k$ describing displacements of $v$ in four linearly independent flat directions. Accordingly, if there are $V$ vertices in $\Sigma_k$ there are $4V$ such constraints $C^k_{vI}$ at step $k$. However, these are not all independent: there are 10 independent global translations and SO(4) rotations which move the entire 3D hypersurface (and the underlying 4D triangulation) in the flat 4D embedding space without affecting the triangulation. Let $E$ be the number of edges in $\Sigma_k$. That is, there exist $A^{vI}_n\neq0$, $m=1,\ldots,10$ such that
\ba\label{globaltrans}
\{y^e_k,A^{vI}_mC^k_{vI}\}=A^{vI}_mY^e_{vI}=0\,,\q\q e=1,\ldots,E\,.
\ea
These ten conditions imply that $\text{rank}(Y^e_{vI})=4V-10$ and, therefore, that there rather exist $4V-10$ linearly independent constraints $C^k_{vI}$ which generate vertex displacements. The (vertex displacement) gauge orbit at step $k$ is therefore $(4V-10)$--dimensional.

The explicit form of the constraints (\ref{CvI}) in terms of areas, angles and lengths is not relevant for this article and will therefore not be exhibited here, however, has been derived in \cite{Dittrich:2009fb} with the help of so-called `tent moves'. It has also been shown in \cite{Dittrich:2009fb} that (a) the `spatial' components $Y^e_{vI}$ are fully determined by the background lengths of the edges in the 3D star of $v$ in $\Sigma_k$,\footnote{The 3D star of $v$ is the collection of all subsimplices in $\Sigma_k$ of dimension 3 or less that share the vertex $v$.} and (b) the contractions $Y_{vI}\cdot N^k$ and $Y_{vI}\cdot M^0$ depend on background variables from $\Sigma_k$ only. Accordingly, the constraints in the first (second) set in (\ref{CvI}) contain linearized as well as background variables from step $k$ (step $0$) only.



Finally, a few comments are in order about the symplectic form which we are working with in the linearized theory. Using the above expansion of the variables, it can be obtained from an expansion of the symplectic form of the full theory $\omega^k=dl^e_k\wedge d{}^+{p}^k_e$ to order $\varepsilon^2$ around a flat background solution, where the ${}^+{\pi}_e^k$ are given in (\ref{reggelinmoms}). Noting that the background variables are fixed, this yields 
\ba\label{linomega}
\delta\omega^k=dy^e_k\wedge d{}^+{\pi}^k_e
\ea
as the symplectic form of the linearized theory. On account of the post--constraints in (\ref{reggecon}), $\delta\omega^k$ is degenerate when restricted to the post--constraint surface.

The effective Lagrangian two--form of the linearized theory can be obtained from $\delta\omega^k$ by pull back under the `effective' discrete Legendre transforms \cite{Dittrich:2013jaa} and similarly reads (the exterior derivative $d$ does not affect the flat background variables)
\ba
\delta\Omega^k:=-\frac{\partial^2\,\tilde{S}_{k}}{\partial l^e_k\partial l^a_0}dy^a_0\wedge dy^e_k=\Omega^k_{ae}\,dy^a_0\wedge dy^e_k\,.\label{lag2formlin}
\ea
The degeneracies of the Lagrangian two--form $\delta\Omega^k$ (\ref{lag2formlin}) of the linearized theory are therefore identical to the ones from the background theory.

\section{Lattice `gravitons' in linearized Regge Calculus}\label{obs}


After discussing the vertex displacement gauge symmetries of linearized Regge Calculus, let us now identify the 
propagating degrees of freedom. Presuming a close link to the continuum and to simplify referring to them, we wish to call the propagating lattice degrees of freedom of linearized Regge Calculus hereafter by the name `gravitons'. However, we emphasize that their relation to the continuum gravitons under a continuum limit is unclear at this stage and we shall also not investigate this relation here.\footnote{This relation is clearly non-trivial because the continuum limit might be achieved via some coarse graining procedure \cite{Dittrich:2014ala,Dittrich:2012jq,Bahr:2009qc,Dittrich:2012qb,Dittrich:2011vz,Bahr:2010cq,Dittrich:2013voa}. This will generically change the dynamical content of the system.}  Nevertheless, we shall see shortly that the lattice `gravitons' correspond to curvature degrees of freedom -- just like their continuum analogues -- and, in fact, provide a linear basis of the  propagating lattice degrees of freedom. The dynamics of these `gravitons' is generated by the evolutions moves with respect to the background discrete time because there are no constraints generating the dynamics -- in contrast to the continuum where the graviton dynamics is generated by a quadratic global Hamiltonian. In this article we shall discuss the canonical Pachner move generated dynamics.\footnote{For tent moves it was shown in \cite{Dittrich:2009fb} that the `gravitons' satisfy discrete second order evolution equations which, in some rough analogy, can be taken as a lattice wave equation.} 

In the present section we shall firstly discuss invariance under vertex displacement gauge symmetries. Being propagating observables, we expect the `gravitons' to be invariant under the action of the constraints $C^k_{vI}$ generating the vertex displacement gauge symmetry and, furthermore, to be associated to curvature. Indeed, by (\ref{Yeps}) we know that the deficit angles are invariant under the vertex displacements in flat directions. Additionally, Barrett's `fundamental theorem of linearized Regge Calculus' \cite{barrett1987fundamental} shows (for topologically trivial triangulated manifolds) that the set of linearized edge length perturbations $y^e$ around a flat background, modulo the subset of linearized length deformations corresponding to vertex displacements in flat directions, is equivalent to the set of {\it linearized deficit angles} (obviously, ${}^{(0)}\epsilon_t=0$)
\ba\label{lineps}
\delta\epsilon_t=\varepsilon\,\frac{\partial \epsilon_t}{\partial l^e}\Big|_\text{flat}y^e+o(\varepsilon^2)\,,
\ea
satisfying the Bianchi identities. In other words, the physical (i.e.\ dynamical) content of linearized Regge Calculus is encoded in the space of linearized deficit angles. Accordingly, the `lattice gravitons' must be closely related to these curvature variables.
However, notice that {\it a priori} $\epsilon_t$ depends on the lengths of all edges in $\text{star}_{4D}(t)$, the 4D star of the bulk triangle $t$ (i.e.\ the collection of simplices which share $t$ as a subsimplex). 

We therefore wish to translate such deficit angles into canonical variables at step $k$ which are invariant under the gauge generators $C^k_{vI}$. To this end, consider a hypersurface $\Sigma_k$ and a bulk triangle $t$ such that $\partial(\text{star}_{4D}(t))\cap\Sigma_k\neq\emptyset$ and the boundary of the 4D star of $t$ touches $\Sigma_k$ `tangentially' (i.e., $\Sigma_k$ does not cut $\text{star}_{4D}(t)$ into two disconnected pieces). Next, we integrate out all edges which are bulk between $\Sigma_0$ and $\Sigma_k$. Employing the equations of motion for the internal edges,
\ba
\frac{\partial^2\,S_{k}}{\partial l^i_k\partial l^{i'}_k}y^{i'}_k+\frac{\partial^2\,S_{k}}{\partial l^i_k\partial l^e_k}y^e_k+\frac{\partial^2\,S_{k}}{\partial l^i_k\partial l^a_0}y^a_0=0\,,
\ea
and making use of (\ref{lineps}), one finds the corresponding linearized `effective' deficit angle
\ba\label{gravsplit}
\delta\tilde{\epsilon}_t=\varepsilon\left(\frac{\partial\, \tilde{\epsilon}_t}{\partial l^e_k}y^e_k+\frac{\partial\,\tilde{\epsilon}_t}{\partial l^a_0}y^a_0\right)+o(\varepsilon^2)\,,
\ea
where in analogy to (\ref{matrixdefs})
\ba
\frac{\partial\, \tilde{\epsilon}_t}{\partial l^e_k}&=&\frac{\partial\,\epsilon_t}{\partial l^e_k}-\frac{\partial\,\epsilon_t}{\partial l^i_k}\,T^i_A\left(T^{i_1}_A\frac{\partial^2\,S_{k}}{\partial l^{i_1}_k\partial l^{i_2}_k}T^{i_2}_{A'}\right)^{-1}T^{i'}_{A'}\,\frac{\partial^2\,S_{k}}{\partial l^{i'}_k\partial l^e_k}\,,\nn\\
\frac{\partial\,\tilde{\epsilon}_t}{\partial l^a_0}&=&\frac{\partial\,\epsilon_t}{\partial l^a_0}-\frac{\partial\,\epsilon_t}{\partial l^i_k}\,T^i_A\left(T^{i_1}_A\frac{\partial^2\,S_{k}}{\partial l^{i_1}_k\partial l^{i_2}_k}T^{i_2}_{A'}\right)^{-1}T^{i'}_{A'}\,\frac{\partial^2\,S_{k}}{\partial l^{i'}_k\partial l^a_0}\,.
\ea
Starting from (\ref{Yeps}), and in precise analogy to (\ref{appomegadeg}), it is straightforward to convince oneself that also the effective deficit angles are invariant under the `spatial' null vectors $Y^e_{vI},Y^a_{vI}$
\ba\label{lindepeps}
Y^e_{vI}\frac{\partial\, \tilde{\epsilon}_t}{\partial l^e_k}=0\,,\q\q Y^a_{vI}\frac{\partial\,\tilde{\epsilon}_t}{\partial l^a_0}=0\,.
\ea
Defining
\ba\label{yt}
y^t_k:=\frac{\partial\, \tilde{\epsilon}_t}{\partial l^e_k}y^e_k\,,
\ea
it is clear that
\ba
\{y^t_k,C^k_{vI}\}=Y^e_{vI}\,\frac{\partial\, \tilde{\epsilon}_t}{\partial l^e_k}=0\q\forall\,v,I\,.
\ea
Hence, those contributions $y^t_k$ of the linearized `effective' deficit angles (\ref{gravsplit}) which depend on data in $\Sigma_k$ constitute non--trivial configuration variables at step $k$ that are invariant under the action of all gauge generators $C^k_{vI}$. These $y^t_k$ admit a geometric interpretation as curvature degrees of freedom and are non--local quantities in that they involve effective expressions obtained after integrating out internal degrees of freedom. We shall see in the sequel that these $y^t_k$ are generally propagating degrees of freedom. Accordingly, we wish to call the $y^t_k$ `gravitons'.

\section{Counting `gravitons' via Pachner moves}\label{sec_gravcount}

Prior to constructing the momenta conjugate to the `gravitons' and analysing their propagation under the Pachner evolution moves, let us count and check whether the $y^t_k$ actually provide a complete set of propagating degrees of freedom. We just verified in section \ref{sec_lincanon} that there are $4V-10$ linearly independent vertex displacement gauge generators $C^k_{vI}$ at step $k$. Hence, if there are $E$ edges in $\Sigma_k$, we should find $E-4V+10$ independent such $y^t_k$ configuration `gravitons', i.e.\ $2(E-4V+10)$ phase space `graviton' degrees of freedom at $k$. We shall show momentarily that this is, indeed, the case for closed $\Sigma_k$.  

Before we do this, a few important comments are necessary to prevent confusion: first of all, as explained in \cite{Dittrich:2013jaa,Hoehn:2014aoa} and summarized in section \ref{sec_can}, the notion of a propagating degree of freedom in the discrete requires {\it two} time steps, say, $0$ and $k$, between which the degree of freedom can propagate. In the case of evolving phase spaces, as necessary for the discretization changing Pachner move dynamics, this strongly depends on these two time steps. 
The number of phase space degrees of freedom propagating from $\Sigma_0$ 
to $\Sigma_k$ reads in the present case \cite{Dittrich:2013jaa,Hoehn:2014aoa}
\ba
N_{0\rightarrow k}&=&2E-2\#(\text{pre--constraints at $0$})\nn\\
&=&2E-2\#(\text{post--constraints at $k$})\nn\\
&=&2E-2(4V-10)-2\#(\text{post--constraints ${}^+C^k$ at $k$ with ${}^+C^k\neq C^k_{vI}$})\nn\,.
\ea

On the other hand, the number $2(E-4V+10)$ at $k$ is clearly {\it independent} of $\Sigma_0$. Consequently, the number of $2(E-4V+10)$ `gravitons' which we are counting at step $k$ does {\it not} necessarily coincide with the `gravitons' that actually propagated from $\Sigma_0$ to $\Sigma_k$ or, likewise, that propagate from $\Sigma_k$ to some $\Sigma_{k+x}$. In particular, the number of post--constraints (or pre--constraints) at step $k$ differing from the gauge generators $C^k_{vI}$ depends, in general, strongly on $\Sigma_0$ (or $\Sigma_{k+x}$). Hence, the number of `gravitons' among the $2(E-4V+10)$ independent ones at $k$ that actually propagated from $\Sigma_0$ to $\Sigma_k$ is generically smaller than $2(E-4V+10)$ (and likewise for propagation from $\Sigma_k$ to $\Sigma_{k+x}$).

As a result, it is more appropriate to view these $2(E-4V+10)$ phase space `graviton' degrees of freedom rather as gauge invariant {\it potentially} propagating degrees of freedom; it depends strongly on initial and final hypersurfaces whether these `gravitons' from hypersurface $\Sigma_k$ actually propagate. But if there are propagating degrees of freedom, they will be contained in this set of `gravitons'. We shall come back to this below.

Let us now return to our attempt to count and show that the `gravitons' $y^t_k$, indeed, form a complete set of gauge invariant degrees of freedom. Viewing $\left(\frac{\partial\, \tilde{\epsilon}_t}{\partial l^e_k}\right)$ as an $E\times N_t$ matrix, where $N_t$ is the total number of bulk triangles $t$ whose $\partial(\text{star}_{4D}(t))$ touches $\Sigma_k$ `tangentially', the first condition in (\ref{lindepeps}) implies 
\ba
\text{rank}\left(\frac{\partial\, \tilde{\epsilon}_t}{\partial l^e_k}\right)\leq E-4V+10\,.\label{rankeps}
\ea 
Let us show that $N_t\geq E-4V+10$ and subsequently that the upper bound in (\ref{rankeps}) is saturated. From this it follows that the $y^t_k$ (\ref{yt}) constitute a complete set.

It is convenient to count the variables by means of the Pachner evolution moves. To this end, we recall the properties of the four Pachner moves from section \ref{sec_pachner} and, in particular, that the deficit angle around the new bulk triangle generated by means of a 2--3 Pachner move is an {\it a priori} free variable. All deficit angles resulting from the 2--3 moves are thus {\it a priori} independent. Denoting by $E_{23}$ the number of edges in $\Sigma_k$ produced by 2--3 moves, it is, therefore, sufficient to show $E_{23}\geq E-4V+10$, which we shall do momentarily. The total number $N_t$ is, of course, generically much larger than $E_{23}$ as a consequence of the 3--2 and 4--1 moves. However, the linearized deficit angles generated during the latter two moves are generally linearly dependent on the deficit angles from the 2--3 moves because there are no new edges introduced in these two types of moves. We will discuss this in more detail in section \ref{sec_linpach} below when studying the linearized Pachner moves. For simplicity, let us assume $\Sigma_k\cap\Sigma_0=\emptyset$ and that the hypersurfaces are closed.


\begin{Proposition}
For any closed triangulated 3D hypersurface $\Sigma_k$ with $\Sigma_k\cap\Sigma_0=\emptyset$ it holds
\ba\label{e23}
E_{23}\geq E-4V+10\,.
\ea
\end{Proposition}
\begin{proof}
Let $\Sigma_k$ be a closed connected hypersurface such that $\Sigma_0\cap\Sigma_k=\emptyset$. Denote by $E_{14}$ the number of edges in $\Sigma_k$ produced through 1--4 moves. It holds $E=E_{14}+E_{23}$ (the 3--2 and 4--1 moves do not introduce new edges). Given that $\Sigma_k$ is closed and $\Sigma_0\cap\Sigma_k=\emptyset$, there must exist some closed hypersurface $\Sigma_{aux}$ (it could be $\Sigma_0$), which does not intersect $\Sigma_k$ but whose vertices are connected to the vertices of $\Sigma_k$. A sequence of Pachner moves can be glued onto $\Sigma_{aux}$ to produce $\Sigma_k$. Since the minimum number of vertices in a closed 3D triangulation is five (the boundary of a 4--simplex), we must glue at least five 1--4 Pachner moves in order to introduce the vertices of $\Sigma_k$. The first of these 1--4 moves must be glued to one tetrahedron with all four of its vertices in $\Sigma_{aux}$ since no other types of Pachner moves placed on $\Sigma_{aux}$ introduce new vertices. The second of the 1--4 moves must be glued to a tetrahedron with at least three vertices in $\Sigma_{aux}$ (the fourth one could be the new one from the previous 1--4 move). Likewise, the third of the 1--4 moves must be glued to a tetrahedron with at least two vertices in $\Sigma_{aux}$ and the fourth of the 1--4 moves must be glued on a tetrahedron with at least one vertex in $\Sigma_{aux}$. Consequently, there are at least $4+3+2+1=10$ edges generated during 1--4 moves, connecting the vertices of $\Sigma_k$ with those of $\Sigma_{aux}$. But these edges must be internal because $\Sigma_{aux}$ does not intersect with $\Sigma_k$. Hence, we conclude that necessarily
\ba
E_{14}\leq 4V-10\,.\nn
\ea
In conjunction with $E=E_{14}+E_{23}$, we thus obtain the desired result.
\end{proof}

Notice that each of the $E_{23}$ edges in $\Sigma_k$ from a 2--3 move is associated to a bulk triangle $t$ with $\partial(\text{star}_{4D}(t))\cap\Sigma_k\neq\emptyset$. That is, we indeed have $N_t\geq E-4V+10$. 

This is sufficient to argue that the bound in (\ref{rankeps}) is saturated for the following reason: if it was not saturated, there must exist a number larger than $4V-10$ of non--trivial vectors $V^e$ such that $V^e\left(\frac{\partial\, \tilde{\epsilon}_t}{\partial l^e_k}\right)=0$, $t=1,\ldots,N_t$. Any such $V^e$ (by (\ref{gravsplit})) defines a transformation in a flat direction. However, the displacements of the vertices in the triangulation already account for all possible transformations in flat directions. In conclusion, there exist precisely $4V-10$ non--trivial null vectors of the $E\times N_t$ matrix $\left(\frac{\partial\, \tilde{\epsilon}_t}{\partial l^e_k}\right)$ and since $N_t\geq E-4V+10$, the bound in (\ref{rankeps}) must be saturated.

That is, among the $N_t$ `gravitons' $y^t_k$ in (\ref{yt}) we can always choose exactly $E-4V+10$ linearly independent ones which we henceforth denote by $y^\alpha_k$, $\alpha=1,\ldots,E-4V+10$. In general, the linearized `effective' deficit angles therefore provide an over--complete set of gauge invariant `gravitons' at step $k$. 


There is a sequence of theorems by Walkup \cite{walkup} concerning characteristic lower bounds for numbers of subsimplices involved in triangulated 3--manifolds which, in our case, ensures that the number of `gravitons' cannot be negative and which can be summarized in the following form \cite{Ambjorn:1996ny}:

\begin{Theorem} \label{thm_walkup}{\bf(Walkup \cite{walkup})}\\
For any combinatorial 3--manifold the inequality
$
E-4V+10\geq0
$
holds with equality if and only if it is a stacked sphere.
\end{Theorem}

A {\it stacked 3--sphere} is a triangulation of the 3--sphere which can be obtained by performing a sequence of 1--4 Pachner moves on the 3D boundary surface of a single 4--simplex. It is not surprising that in this case the total number of `gravitons' must be zero, since the 1--4 moves only generate further boundary data, but do not introduce any internal triangles and thus also no curvature.

Note, however, that the configuration of a stacked sphere can also be obtained as the 3D boundary surface of a 4D triangulation involving internal triangles and edges. Nevertheless, any triangulation whose boundary configuration corresponds to a stacked sphere possesses flat solutions independently of the existence of internal triangles. The reason is that the boundary data of a stacked sphere arising from a single simplex and only 1--4 Pachner moves can be freely chosen (modulo generalized triangle inequalities) without changing flatness. In particular, it can be chosen to coincide with the boundary data of any triangulation whose boundary corresponds to a stacked sphere (of the same connectivity), yet which possesses internal edges and triangles.\footnote{In fact, there may exist special curved solutions as well for boundary data otherwise admitting flatness. However, in contrast to the flat solutions, these curved solutions are isolated in that there exists no continuous symmetry of the solutions and rather seem to constitute a discretization artifact \cite{Bahr:2009ku}. We shall ignore here such special isolated solutions.} Any stacked sphere configuration (involving internal triangles or not) is fully constrained in the canonical formalism with as many constraints as edges in the boundary (see also \cite{Dittrich:2011ke}).

\begin{Example}\label{ex_1}
\emph{Consider a single 4--simplex and perform five 1--4 gluing moves on the five tetrahedra of the boundary. Subsequently, carry out five further 4--1, ten 2--3 and ten 3--2 Pachner moves, in order to obtain a new 3D boundary which does not intersect with the boundary of the original 4--simplex. This is always possible. Thus, the 10 edges connecting the five new vertices in this surface all resulted from the ten 2--3 Pachner moves (the 3--2 and 4--1 moves do not introduce new edges). However, the new boundary configuration is identical to the one of a single 4--simplex and is, therefore, a stacked sphere with $E-4V+10=0$ `gravitons', despite the fact that all edges are associated to internal triangles. This is only possible if all contributions from the linearized deficit angles around the internal triangles vanish. Indeed, both $\frac{\partial\,\tilde{\epsilon}^\alpha}{\partial l^e}$ and $Y^e_{vI}$ are $10\times10$ matrices, where the latter is non--degenerate due to the presence of 10 linearly independent directions for displacing the 5 vertices in 4D flat space. Hence, (\ref{lindepeps}) implies $\frac{\partial\,\tilde{\epsilon}^\alpha}{\partial l^e}=0$ such that this configuration is devoid of `gravitons'. 
Clearly, the present configuration is a fully constrained one with 10 linearly independent vertex displacement generators $C_{vI}$ and 10 edges in the boundary. 
 }
\end{Example}

We stress that this example does {\it not} imply that every spherical triangulation or every fully constrained configuration is flat and devoid of `gravitons'. For instance, the fact that the phase space at a given time step in the evolution is fully constrained reflects the fact that no degrees of freedom propagate {\it through} that time step but it does {\it not} necessarily imply that no degrees of freedom propagate away from this step in only either `future' or `past' direction. This subtle issue has been thoroughly discussed in \cite{Dittrich:2013jaa,Hoehn:2014aoa} -- also in the context of a discrete version of the `no boundary' proposal -- and we shall not repeat it here. It roots in the fact that in the discrete one always has to specify {\it two} time steps in order to discuss propagation and different pairs of time steps on a temporally varying discretization may yield different numbers of propagating degrees of freedom between them.

The example \ref{ex_1} of a stacked sphere (with internal edges) shows that it is possible that $E=E_{23}$. Hence, in combination with Walkup's theorem \ref{thm_walkup}, we have shown:
\begin{Theorem}
For any closed 3D hypersurface $\Sigma_k$ with $\Sigma_k\cap\Sigma_0=\emptyset$ the sequence of inequalities
\ba
E\geq E_{23}\geq E-4V+10\geq0
\ea
holds with equality in the last relation if and only if it is a stacked sphere.
\end{Theorem}

From the considerations thus far, we can already predict the role played by each of the Pachner moves in the `generation' and `annihilation' of `gravitons' and gauge degrees of freedom on the {(non--extended) evolving phase spaces}. At the configuration space level, the number of `gravitons' at $k$ is $E-4V+10$, while it follows from the analysis in \cite{Dittrich:2013jaa,Hoehn:2014aoa} that to each of the $4V-10$ independent vectors $Y_{vI}$ there is associated one gauge degree of freedom (we shall exhibit these gauge variables more explicitly below in section \ref{sec_tmatrix}). Denote the changes in the number of edges and vertices when going from $\Sigma_k$ to $\Sigma_{k+1}$ by means of any of the Pachner moves by $\Delta E$ and $\Delta V$, respectively. Compute the net changes in the numbers of `gravitons' as $\Delta N_p=\Delta E-4\Delta V$ and the net changes in the numbers of gauge variables via $\Delta N_g=4\Delta V$. Using the geometric properties of the Pachner moves exhibited in section \ref{sec_pachner}, we conclude:
\begin{description}
\item[1--4 Pachner move:] generates 4 new gauge modes: $\,\Delta V=+1$, $\,\,\Delta E=+4$ $\Rightarrow$ $\Delta N_p=0$ and $\Delta N_g=+4$
\item[2--3 Pachner move:] generates 1 new `graviton': $\,\q\,\,\,\Delta V=0$, $\,\,\,\,\,\,\Delta E=+1$ $\Rightarrow$ $\Delta N_p=+1$ and $\Delta N_g=0$
\item[3--2 Pachner move:] annihilates 1 old `graviton': $\,\q\,\,\Delta V=0$, $\,\,\,\,\,\Delta E=-1$ $\Rightarrow$ $\Delta N_p=-1$ and $\Delta N_g=0$
\item[4--1 Pachner move:] annihilates 4 old gauge modes: $\Delta V=-1$, $\Delta E=-4$ $\Rightarrow$ $\Delta N_p=0$ and $\Delta N_g=-4$
\end{description}

We emphasize, again, that the 
`gravitons' at each step $k$ are 
invariant under the vertex displacement gauge symmetry but only {\it potentially} propagate from or to $\Sigma_k$. 
For instance, the fact that the
\begin{description}
\item[2--3 move] evolving $\Sigma_k$ to $\Sigma_{k+1}$ `generates' a `graviton' is to be understood in the following sense: the new `graviton' of the 2--3 move is an {\it a priori} free variable at $k+1$ that cannot be predicted by the data at $k$. It, therefore, does {\it not} propagate {\it to} $\Sigma_{k+1}$. However, it {\it may} propagate from $\Sigma_{k+1}$ onwards to some $\Sigma_{k+x}$, $x>1$, depending on the particulars of the `future' triangulation. 
\item[3--2 move] from $\Sigma_k$ to $\Sigma_{k+1}$ `annihilates' a `graviton' means that the number of (configuration) degrees of freedom {\it potentially} propagating to or from $\Sigma_{k+1}$ is decreased by one as compared to $\Sigma_k$. The removed `graviton' is {\it a posteriori} free at $k$. Whether it actually propagated to $\Sigma_k$ 
depends on the `past' triangulation. 
\end{description}
We shall discuss these linearized Pachner move dynamics further in canonical language in section \ref{sec_linpach}.


\section{Disentangling gauge and `graviton' phase space variables}\label{sec_tmatrix}

It is convenient to perform linear canonical transformations on the linearized pairs $(y_k,\pi^k)$ in order to disentangle the gauge from the `graviton' degrees of freedom. We shall follow here the more general procedure developed in \cite{Hoehn:2014aoa} for arbitrary quadratic discrete actions. While in the general case one can distinguish among eight different types of degrees of freedom, it suffices for our purposes to only distinguish 
between the two broad types given by the `gravitons' as defined above and the gauge variables conjugate to the constraints. (It is the `gravitons' for which one could distinguish further cases, according to whether they actually propagate to or from $\Sigma_k$ or both or not at all \cite{Hoehn:2014aoa}.) 

We proceed as follows: choose $4V-10$ linearly independent displacement vectors $(Y_k)^e_{vI}$, $vI=1,\ldots,4V-10$ (henceforth we include a time step label $k$) and $E-4V+10$ linearly independent $\frac{\partial\, \tilde{\epsilon}^\alpha}{\partial l^e_k}$, $\alpha=1,\ldots,E-4V+10$. We construct an invertible transformation matrix $(T_k)^e_\Gamma$, where the index set $\Gamma$ runs over both $vI$ and $\alpha$, by firstly setting
\ba\label{Tcond2}
(T_k)^e_{vI}=(Y_k)^e_{vI}
\ea
and 
\ba\label{Tcond1}
(T_k^{-1})^\alpha_e=\frac{\partial\, \tilde{\epsilon}^\alpha}{\partial l^e_k}\,.
\ea
Clearly, $(T_k)^e_{vI}(T_k^{-1})^\alpha_e=0$. Next, we choose $E-4V+10$ linearly independent $(T_k)^e_\alpha$ such that
\ba
(T_k)^e_\alpha(T_k^{-1})^{\beta}_e=\delta^\beta_\alpha\,.
\ea
Certainly, these conditions to {\it not} uniquely determine the matrix $(T_k)^e_\Gamma$. However, any choice satisfying the above conditions is sufficient for our purposes. Assume, therefore, that such a choice has been made at step $k$. 

Notice that each such $(T_k)_\alpha$ defines a variation of the edge lengths in $\Sigma_k$ such that only a single (independent) `effective' deficit angle is changed, since by construction $(T_k)^e_\alpha\frac{\partial\,\tilde{\epsilon}^\beta}{\partial l^e_k}=\delta^\beta_\alpha$. That is, in contrast to the $4V-10$ $(Y_k)_{vI}$ which leave the geometry invariant, the $E-4V+10$ $(T_k)_\alpha$ actually define geometry changing directions. As a consequence, these $(T_k)_\alpha$ will generically {\it not} define degenerate directions of (effective) Hessians such as (\ref{p2SR}).\footnote{We cannot preclude, in general, that there exist geometry changing directions $(T_k)_\alpha$ which, nevertheless, are null vectors of the Hessian, $(T_k)^e_\alpha H_{ee'}=0$, and leave the Regge action (expanded to second order around the flat background) invariant. The accompanying transformation would need to generate geometry changes that lead to variations of the action which cancel each other; this could only occur in special situations.} We therefore choose to label the non--degenerate directions of the Hessian by $\alpha$. We shall see that the $(T_k)_\alpha$ may still define degenerate directions of the Lagrangian two--forms.

Using this transformation matrix $(T_k)^e_\Gamma$, 
we perform a linear canonical transformation \cite{Hoehn:2014aoa}
\ba
y^\Gamma_k=(T_k^{-1})^\Gamma_e\,y^e_k\,,\q\q\q\q p^k_\Gamma=(T_k)^e_\Gamma\,\pi^k_e\,,\q\q\q\q e=1,\ldots,E\,.\label{variablesplit}
\ea
(For notational simplicity and assuming momentum matching, we henceforth drop the ${}^+,{}^-$ at the momenta.) In particular, we now have the $E-4V+10$ `gravitons'
\ba
y^\alpha_k=(T_k^{-1})^\alpha_ey^e_k=\frac{\partial\, \tilde{\epsilon}^\alpha}{\partial l^e_k}y^e_k\,.\nn
\ea
However, their conjugate momenta $p^k_\alpha$ thus defined are generally not invariant under the vertex displacement gauge symmetry generated by (\ref{CvI}) because
\ba\label{lala}
\{p^k_\alpha,C^k_{vI}\}=-(T_k)^e_\alpha N^k_{ee'}(Y_k)^{e'}_{vI}
\ea
may generally not vanish. As in \cite{Hoehn:2014aoa}, it is thus useful to perform a second linear canonical transformation. 

Beforehand, let us simplify the notation for the sequel and  define the following transformed matrices
\ba
\Omega^k_{\alpha\beta}&:=&(T_0)^e_\alpha \Omega^k_{ae}(T_k)^{e}_\beta\,,\nn\\
N^k_{\alpha\beta}&:=&(T_k)^e_\alpha N^k_{ee'}(T_k)^{e'}_\beta\,,\nn\\
N^k_{\alpha vI}&:=&(T_k)^e_\alpha N^k_{ee'}(Y_k)^{e'}_{vI}\,,\nn\\
N^k_{vIwJ}&:=&(Y_k)^e_{vI}N^k_{ee'}(Y_k)^{e'}_{wJ}\,,\label{Tcond3}
\ea
where $\Omega^k_{ae},N^k_{ee'}$ are given in (\ref{reggelinmoms}, \ref{matrixdefs}). 





To cleanly disentangle the gauge from the `graviton' variables, 
we carry out a second canonical transformation. Obviously, there exist many possible choices for such transformations. We choose one which leaves the configuration data $y^{vI}_k,y^\alpha_k$ invariant,
\ba
y^{vI}_k&\rightarrow&y^{vI}_k\,,\q\q\q p^k_{vI}\,\,\rightarrow\,\,\pi^k_{vI}:=(Y_k)^e_{vI}\,\pi^k_e-N^k_{vI\alpha}\,y^\alpha_k\,.\nn\\
y^{\alpha}_k&\rightarrow&y^{\alpha}_k\,,\q\q\q\,\,\,\,
p^k_{\alpha}\,\,\rightarrow\,\,\pi^{k}_\alpha\,\,:=(T_k)^e_\alpha\,\pi^{k}_e-N^{k}_{\alpha vI}\,y^{vI}_{k}\,.\label{ncantrans}
\ea
It is straightforward to check that this defines a canonical transformation with new canonical pairs $(y^{vI}_k,\pi^k_{vI})$ and $(y^\alpha_k,\pi^k_\alpha)$. In particular, the vertex displacement generators (\ref{CvI}) now appear in a form devoid of `graviton' variables
\ba\label{14purecons}
C^k_{vI}=\pi^k_{vI}-N^k_{vIwJ}\,y^{wJ}_k\,.
\ea
The $y^{vI}_k$, therefore, constitute the $4V-10$ gauge variables conjugate to the constraints. We may interpret them as four linearized `lapse and shift' variables associated to each vertex in $\Sigma_k$. Their conjugate momenta $\pi^k_{vI}$ are constrained by (\ref{14purecons}). Furthermore, we have 
\ba
\{y^\alpha_k,C^k_{vI}\}=0\,,\q\q\q\q\{\pi^k_\alpha,C^k_{vI}\}=0\,,\nn
\ea
such that $(y^\alpha_k,\pi^k_\alpha)$ form $E-4V+10$ canonical `graviton' pairs which are invariant under the vertex displacement symmetry of linearized Regge Calculus. 

It is desirable to have a proper geometric interpretation of the `graviton' momenta $\pi^k_\alpha$. For example, in linearized continuum general relativity, the gauge invariant momenta conjugate to the graviton variables are related to the extrinsic curvature. In the present discrete formulation it is a bit more difficult to give a precise geometric interpretation in terms of extrinsic geometry, in part because (1) clearly the $\pi^k_\alpha$ are not unique once the `gravitons' $y^\alpha_k$ are chosen, and (2) the matrices $(T_k)^e_\alpha$ and $N^k_{\alpha vI}$ cloud the information about the background geometry. For the time being, we leave the geometric interpretation of the `graviton' momenta as an open problem and content ourselves with the observation that the $\pi^k_\alpha$ are generators of geometry (i.e.\ linearized deficit angle) changing transformations $\{y^\alpha_k,\pi^k_\beta\}=\delta^\alpha_\beta$.

Let us now study their dynamics as generated by the Pachner moves.

\section{Pachner moves in 4D linearized Regge Calculus}\label{sec_linpach}

The Pachner moves locally evolve the `spatial' hypersurface $\Sigma$ forward in discrete time. When integrating out any new internal edges produced by these local evolution moves, the sequence of Pachner moves $k\rightarrow k+1$, $k+1\rightarrow k+2$, ... is equivalent to the sequence of {\it global} evolution moves $0\rightarrow k$, $0\rightarrow k+1$, $0\rightarrow k+2$, ... In other words, the Pachner moves update the global evolution moves and one implicitly considers the propagation of data from $\Sigma_0$ onto the evolving hypersurface $\Sigma_k$. Since the number of `gravitons' propagating from $\Sigma_0$ to $\Sigma_k$ cannot increase with growing $k$, the rank of the symplectic form on the evolving slice can only remain constant or decrease. Indeed, theorem 4.1 in \cite{Dittrich:2013jaa} shows for full non-perturbative Regge Calculus that the 1--4 and 2--3 Pachner moves preserve the symplectic form restricted to the post--constraint surfaces, whereas the 3--2 and 4--1 moves can reduce the rank of the symplectic form restricted to the post--constraint surface (depending on the specifics of the Regge triangulation). 


In section \ref{sec_gravcount} we have already revealed the general role assumed by each of the Pachner moves in the evolution. We shall study these roles now in detail. In particular, we shall see that the 1--4 moves generate the vertex displacement gauge generators (\ref{CvI}) and gauge variables $x^{vI}_k$ at each vertex which then are preserved by all Pachner moves until a 4--1 move renders the corresponding vertex internal and trivializes the associated constraints. The preservation of the four constraints per vertex in each hypersurface has to be expected as we will always work on solutions which inherit the gauge symmetry of the background. In fact, it directly follows from the general theorem 4.1 in \cite{Dittrich:2013jaa}. Nevertheless, we shall use it as an important consistency check of the formalism. In addition, the 2--3 moves `generate' gravitons, while the 3--2 moves `annihilate' them. As the Pachner moves constitute an elementary and ergodic set of evolution moves from which {\it all} other (topology preserving) evolution moves can be constructed, the end product of the present section will be a completely general account of the linearized (canonical) dynamics of Regge Calculus in 4D. 
For better readability we shall move technical details to appendix \ref{app_linpach}.

The decomposition of the canonial transformation matrix $(T_k)^e_\Gamma$ of section \ref{sec_tmatrix} changes on solutions to equations of motion because, in particular, the coarse graining or lattice shrinking evolution moves change the dynamical content of the system \cite{Hoehn:2014aoa}. After all, the classification of degrees of freedom is spacetime region dependent for temporally varying discretizations. It will therefore not come as a great surprise that we will not only have to extend or reduce, but also transform this matrix along the way of the Pachner evolution---after all, each $\Sigma_k$ is equipped with a different set of degrees of freedom. Specifically, a non-trivial transformation of $(T_k)^e_\Gamma$ will happen during the 3--2 moves which provide the {\it only} non--trivial equations of motion of the linearized theory. 

Finally, as regards notation: 
we continue to label {\it all} edges in $\Sigma_k\cap\Sigma_{k+1}$ by $e$, while, newly introduced edges are indexed by $n$ and old edges which are rendered internal are labeled by $o$. The linearized Pachner moves involve second derivatives of the action $S_\sigma$ of the newly glued 4-simplex $\sigma$. To simplify the notation in the sequel, we define:

%
%
%
%
%
\ba
S^\sigma_{ee'}&:=&\,\,\,\frac{\partial^2\,S_\sigma}{\partial l^e_k\partial l^{e'}_k}\,,\nn\\
S^\sigma_{\alpha\beta}&:=&(T_k)^e_\alpha S^\sigma_{ee'}(T_k)^{e'}_\beta\,,\nn\\
S^\sigma_{\alpha vI}&:=&(T_k)^e_\alpha S^\sigma_{ee'}(Y_k)^{e'}_{vI}\,,\nn\\
S^\sigma_{vIwJ}&:=&(Y_k)^e_{vI}S^\sigma_{ee'}(Y_k)^{e'}_{wJ}\,.\nn
\ea

\subsection{The `linearized' 1--4 Pachner move}\label{sec_14}

A 1--4 Pachner move performed on $\Sigma_k$ yields four new edges labeled by $n$, a new vertex $v^*\subset\Sigma_{k+1}$ and the corresponding four new `lapse and shift' gauge variables enumerated by $v^*I$, but does not introduce new bulk triangles (see figure \ref{14m}). The momentum updating (\ref{4c6}, \ref{4c6c}) reads in linearized form
\ba\label{lin14}
y^e_{k+1}&=&y^e_k\,,\q\q\pi^{k+1}_e=\pi^k_e+S^\sigma_{ee'}\,y^{e'}_{k+1}+S^\sigma_{en}\,y^{n}_{k+1}\,,\nn\\
\pi^k_n&=&0\,,\q\q\,\,\,\pi^{k+1}_n=S^\sigma_{nn'}\,y^{n'}_{k+1}+S^\sigma_{ne}\,y^{e}_{k+1}\,.
\ea
There are no linearized equations of motion involved in this move. Notice that the very last equation constitutes four post--constraints, one for each edge $n$.

Since $y^e_{k+1}=y^e_k$ and clearly $y^\alpha_{k+1}=y^\alpha_k$ (the previous deficit angles do not change under the addition of a new boundary simplex), we would also like to maintain the same linearized `lapse and shift' for the old vertices, $y^{vI}_{k+1}=y^{vI}_k$ ($y^n_{k+1}$ are not needed in order to determine the embedding of the old vertices). 

In appendix \ref{app_14}, using a suitable extension at step $k+1$ of the transformation matrix $(T_k)^e_\Gamma$, we show:
\begin{itemize}
\item[(i)] The vertex displacement generators of all already existing $v\subset\Sigma_k\cap\Sigma_{k+1}$ are preserved under (\ref{lin14}) as
\ba
C^k_{vI}=\pi^k_{vI}-N^k_{vIwJ}\,y^{wJ}_k=0\q\q\q\Rightarrow\q\q\q C^{k+1}_{vI}=\pi^{k+1}_{vI}-N^{k+1}_{vI\tilde{v}J}\,y^{\tilde{v}J}_{k+1}=0\,,\nn
\ea
where $\tilde{v}$ now includes both $v$ and the new $v^*$.

\item[(ii)] Similarly, the four new post--constraints in the last equation of (\ref{lin14}) become, upon contraction with the new $Y_{v^*I}$, the vertex displacement generators of the new vertex $v^*\subset\Sigma_{k+1}$. The latter read
\ba
C^{k+1}_{v^*I}:=\pi^{k+1}_{v^*I}-S^\sigma_{v^*I\tilde{v}J}\,y^{\tilde{v}J}_{k+1}\label{a14momgaup2}
\ea
and are conjugate to the four new `lapse and shift' gauge degrees of freedom $y^{v^*I}_{k+1}$ associated to $v^*$.

\item[(iii)] No new `gravitons' are produced. The momenta of the existing `gravitons' evolve under (\ref{lin14}) as
\ba
\pi^{k+1}_\alpha=\pi^k_\alpha+S^\sigma_{\alpha\beta}\,y^\beta_{k}\,.\nn
\ea
These are generally proper evolution equations and will be called {\it `graviton' momentum updating}. 
\end{itemize}

As before, the sets  $(y^{\tilde{v}I}_{k+1},\pi^{k+1}_{\tilde{v}I})$ and $(y^\alpha_{k+1},\pi^{k+1}_\alpha)$ are canonically conjugate pairs of gauge variables and gauge invariant `graviton' degrees of freedom, respectively.

\subsection{The `linearized' 2--3 Pachner move}\label{sec_23}

Performing a 2--3 Pachner move on $\Sigma_k$ introduces one new edge labeled by $n$ and a new {\it a priori} free deficit angle $\epsilon^{\alpha^*}$ (see figure \ref{23m}). The linearized momentum updating (\ref{4c6}, \ref{4c6c}) corresponding to the 2--3 move is in shape identical to (\ref{lin14})---with the sole difference that $n$ now labels only one new edge. Accordingly, no linearized equation of motion arises (no new internal edge is created) and there is now one post--constraint labeled by $n$. Again, $y^e_{k+1}=y^e_k$, $y^\alpha_{k+1}=y^\alpha_k$ and we also choose to keep $y^{vI}_{k+1}=y^{vI}_k$.

We show in appendix \ref{app_23}, using a suitable extension of the transformation matrix $(T_k)^e_\Gamma$, that:

\begin{itemize}
\item[(i)] The vertex displacement generators of all $v\subset\Sigma_k\cap\Sigma_{k+1}$ are preserved under the 2--3 Pachner move
\ba
C^k_{vI}=\pi^k_{vI}-N^k_{vIwJ}y^{wJ}_k=0\q\q\q\Rightarrow\q\q\q C^{k+1}_{vI}=\pi^{k+1}_{vI}-N^{k+1}_{vI{w}J}y^{{w}J}_{k+1}=0\,.\label{a23newcon4}
\ea

\item[(ii)] The momenta of the old `gravitons' evolve according to the {\it `graviton' momentum updating}
\ba
\pi^{k+1}_\alpha=\pi^k_\alpha+S^\sigma_{\alpha\tilde{\beta}}\,y^{\tilde{\beta}}_{k}\,\nn
\ea
under the 2--3 move, where $\tilde{\alpha}$ now runs over the old `gravitons' $\alpha$ and the new $\alpha^*$.

\item[(iii)] The momentum conjugate to the newly generated `graviton' $y^{\alpha^*}_{k+1}$ reads
\ba\label{23momobupnew}
\pi^{k+1}_{\alpha^*}=S^\sigma_{\alpha^* \tilde{\beta}}\,y^{\tilde{\beta}}_{k+1}\,.
\ea
This is the post--constraint of the 2--3 move of the linearized theory and a consequence of the vector $(T_{k+1})_{\alpha^*}$ being a right null vector of the effective Lagrangian two-form $\Omega^{k+1}$ at step $k+1$.
\end{itemize}

In contrast to the four new post--constraints (\ref{a14momgaup2}) produced during the 1--4 Pachner move which are generators of the vertex displacement gauge symmetry and only contain gauge variables and their conjugate momenta, the single post--constraint (\ref{23momobupnew}) of the 2--3 move constrains the momentum of the new `graviton' and generically does not constitute a gauge generator. It manifests the fact that the new `graviton' $y^{\alpha^*}_{k+1}$ is an {\it a priori} free variable that cannot be predicted by the data on $\Sigma_k$ (or $\Sigma_0$), i.e.\ that this `graviton' $y^{\alpha^*}_{k+1}$ did {\it not} propagate from $\Sigma_{k}$ (or $\Sigma_0$) to $\Sigma_{k+1}$. However, if $y^{\alpha^*}_{k+1}$ does not turn out to be also {\it a posteriori} free, it may propagate from $\Sigma_{k+1}$ onwards in an evolution $k+1\rightarrow k+x$ specified by initial data at $k+1$ (see also section \ref{sec_gravcount}). The post--constraint (\ref{23momobupnew}) can be interpreted as a {\it refinement consistency condition} which ensures that the fewer data of the `coarser' hypersurface $\Sigma_k$ can be consistently embedded in the larger phase space of the `finer' $\Sigma_{k+1}$.

 Finally, we note that $(y^{vI}_{k+1},\pi^{k+1}_{vI})$, $(y^{\tilde{\alpha}}_{k+1},\pi^{k+1}_{\tilde{\alpha}})$, again, define canonically conjugate pairs of gauge and `graviton' variables, respectively; all $(y^{\tilde{\alpha}}_{k+1},\pi^{k+1}_{\tilde{\alpha}})$, $\tilde{\alpha}=\alpha,\alpha^*$, Poisson commute with the vertex displacement gauge symmetry generators (\ref{a23newcon4}).

\subsection{The `linearized' 3--2 Pachner move}\label{sec_32}

A 3--2 Pachner move renders an old edge labeled by $o$ internal (see figure \ref{23m}) and, as we shall see shortly, removes a `graviton' which we label by $\alpha^*$ (the 3--2 move does not affect the number of vertices). This move causes quite a bit more trouble. The momentum updating (\ref{c14b}, \ref{c14c}) of the 3--2 move reads as follows in linearized form:
\ba\label{lin32}
y^e_{k+1}&=&y^e_k\,\q\q\q\pi^k_e=\pi^{k+1}_e-S^\sigma_{ee'}\,y^{e'}_k-S^\sigma_{eo}\,y^{o}_k\,,\nn\\
\pi^{k+1}_o&=&0\,,\q\q\q\pi^{k}_o=-S^\sigma_{oo'}\,y^{o'}_k-S^\sigma_{oe}\,y^{e}_k\,.
\ea
The very last equation constitutes the linearized pre--constraint of the 3--2 Pachner move. Using (\ref{reggelinmoms}) and (\ref{matrixdefs}) for $\pi^k_o$, we can write it as the equation of motion of the new internal edge labeled by $o$,
\ba\label{32geom}
\left(N^k_{oo}+S^\sigma_{oo}\right)\,y^o_{k}+\left(N^k_{oe}+S^\sigma_{oe}\right)\,y^e_{k}-\Omega^k_{ao}\,y^a_0=0\,.
\ea
In appendix \ref{app_32} it is shown that this equation of motion can be written entirely in terms of `gravitons'
\ba\label{a32gravlindep}
\left(N^k_{o\alpha^*}+S^\sigma_{o\alpha^*}\right)\,y^{\alpha^*}_k+\left(N^k_{o\alpha}+S^\sigma_{o\alpha}\right)\,y^{\alpha}_k+\Omega^k_{\gamma o}\,y^\gamma_0=0.
\ea
($y^\gamma_0$ is a `graviton' on $\Sigma_0$.) This pre--constraint of the 3--2 Pachner move yields the only non-trivial equation of motion for linearized Regge Calculus (the linearized equations of motion of the 4--1 move below are automatically satisfied). Being a pre--constraint, it can lead to the following situations:\footnote{It is difficult to preclude that a third situation may occur in which the pre--constraint is dependent on the post--constraints and thus automatically satisfied. However, in this case, the constraint would be a gauge generator \cite{Hoehn:2014aoa} despite only involving curvature degrees of freedom. This case is not plausible and we shall therefore assume it not to occur.}
\begin{itemize}
\item[(a)] It is independent of the post--constraints at $k$ and does {\it not} fix any of the {\it a priori} free data at $k$. That is, it is first class and restricts the space of solutions (space of initial data) leading to $\Sigma_k$.\footnote{In this case the 3--2 move reduces the rank of the symplectic form restricted to the post--constraint surface by two in the evolution from $k$ to $k+1$ \cite{Dittrich:2013jaa}.} The pre--constraint prevents one (configuration) `graviton' that propagated to $\Sigma_k$ from propagating further to $\Sigma_{k+1}$. This can be seen from (\ref{a32gravlindep}) which {\it cannot} contain any {\it a priori} free `graviton' in this case (otherwise it would become fixed) and thus `annihilates' one independent propagating `graviton' at $k$ by linear dependence with the others. 
\item[(b)] It is independent of the post--constraints but fixes one {\it a priori} free datum via (\ref{a32gravlindep}) which thus must be an {\it a priori} free `graviton'. In this case, the pre--constraint will be second class\footnote{Accordingly, it cannot further reduce the rank of the symplectic form restricted to the post--constraint surface \cite{Dittrich:2013jaa}.} and {\it not} prevent an actually propagating `graviton' from propagating further to $\Sigma_{k+1}$ because only an {\it a priori} free `graviton' that did {\it not} propagate to $\Sigma_k$ gets fixed. (Recall that the set of $E-4V+10$ {\it potentially} propagating `gravitons' at each step may contain {\it a priori} free modes.) 
\end{itemize}


In both cases, the equation of motion (\ref{a32gravlindep}) will require a non-trivial transformation of the matrix when going from $(T_k)$  to $(T_{k+1})$ as shown in appendix \ref{app_32}. We may choose $y^{\alpha^*}_k$ to be the `graviton' that either gets `annihilated' as in (a) or fixed as in (b). This change of canonical transformation matrix at $k+1$ entails a shift in the remaining `lapse and shift' and `graviton' variables
\ba
y^{vI}_{k+1}=(T_{k+1}^{-1})^{vI}_e\,y^e_{k+1}\neq y^{vI}_k \,,\q\q  y^\alpha_{k+1}=(T_{k+1}^{-1})^\alpha_e\,y^e_{k+1}\neq y^\alpha_k\,,\nn
\ea 
despite $y^e_{k+1}=y^e_k$. From step $k+1$ onwards we will employ $y^\alpha_{k+1}$ and $y^{vI}_{k+1}$ as `graviton' and gauge variables, respectively. As mentioned earlier, this does not come as a great surprise because the 3--2 move can be considered as a coarse graining or lattice shrinking move which changes the dynamical content of the system at a given time step. The pre--constraint (\ref{a32gravlindep}) may be interpreted as a {\it coarse graining consistency condition}, ensuring that the larger amount of dynamical data on the `finer' $\Sigma_k$ can be mapped to the smaller phase space of the `coarser' $\Sigma_{k+1}$, thereby reducing the amount of dynamical data at $k+1$. Ultimately, on each hypersurface $\Sigma_k$ we consider a different set of degrees of freedom: at step $k+1$ we now consider the propagation $0\rightarrow k+1$ and no longer $0\rightarrow k$. 

The necessity for the shift in the `graviton' modes may also be directly seen from (\ref{gravsplit}) which gives the linearized `effective' deficit angles: the `gravitons' $y^t_k$ were only those contributions from these effective deficit angles that depend on the data at step $k$. When some of these data, in this case $y^o_k$, becomes internal, the corresponding equation of motion shifts part of the contribution to $y^t_k$ from $\frac{\partial\,\tilde{\epsilon}_t}{\partial l^o_{k}}y^o_{k}$ to $\frac{\partial\,\tilde{\epsilon}_t}{\partial l^a_0}y^a_0$ and $\frac{\partial\,\tilde{\epsilon}_t}{\partial l^e_{k}}y^e_{k}$. That is, after $y^o_k$ has been integrated out, the contribution from the data in $\Sigma_{k+1}$ to the effective deficit angles has shifted and, accordingly, the `graviton' modes become shifted too. On the other hand, it is also clear that the contributions of the various edges to the $y^{vI}$, i.e.\ the $(T^{-1})^{vI}_c$, must be transferred to different edges in the course of the Pachner move evolution since all edges which initially determined the embedding of the vertex (e.g.\ after a 1--4 move) may become internal before the vertex itself is rendered internal. It is therefore neither surprising that also the gauge modes---corresponding to the embedding coordinates of the vertices---experience a shift.

Furthermore, using the transformed $(T_{k+1})$ at $k+1$, it is demonstrated in appendix \ref{app_32} that:
\begin{itemize}
\item[(i)] The vertex displacement generators of all $v\subset\Sigma_k\cap\Sigma_{k+1}$ are preserved under the 3--2 Pachner move
\ba
C^k_{vI}=\pi^k_{vI}-N^k_{vIwJ}\,y^{wJ}_k=0\q\q\q\Rightarrow\q\q\q C^{k+1}_{vI}=\pi^{k+1}_{vI}-N^{k+1}_{vI{w}J}\,y^{{w}J}_{k+1}=0\,\nn
\ea

\item[(ii)] The pre--constraint of the 3--2 move trivializes into the momentum conjugate to the `annihilated' or fixed `graviton'
\ba
\pi^{k+1}_{\alpha^*}=0\,.\nn
\ea

\item[(iii)] On account of the change in the transformation matrix when going from $k$ to $k+1$ one cannot write the `graviton' momentum updating in the simple form as for the 1--4 and 2--3 moves. The equations taking its place are not illuminating and will not be reproduced here. Instead, the remaining `graviton' momenta can also be written as
\ba
\pi^{k+1}_\alpha=N^{k+1}_{\alpha\beta}\,y^\beta_{k+1}-\Omega^{k+1}_{\gamma\alpha} \,y^\gamma_0\,.\nn
\ea
That is, through the non--trivial equations of motion of the 3--2 move, the graviton momenta generally depend on the initial data---in contrast to the momenta conjugate to the gauge modes which are just constrained.

\end{itemize}

Similarly to before, also the shifted variables $(y^{vI}_{k+1},\pi^{k+1}_{vI})$ and $(y^\alpha_{k+1},\pi^{k+1}_\alpha)$ yield a canonically conjugate set of gauge and `graviton' degrees of freedom for step $k+1$.

Finally, one may wonder whether the three new bulk triangles generated during the 3--2 move yield any new `gravitons'. Indeed, these lead to new linearized `effective' deficit angles (\ref{yt}) and therefore to `gravitons'. However, these are linearly dependent on the ones already present at step $k$: it follows from section \ref{sec_gravcount} that the rank of the matrix $\f{\p\tilde{\epsilon}^\alpha}{\p l^e_{k+1}}$ is $(E-1)-4V+10$ after the 3--2 move because the number of vertices did not change. Since we had $E-4V+10$ independent such `gravitons' at step $k$ and there is only one non--trivial pre--constraint (\ref{a32gravlindep}) in the move, the old set of `gravitons' is simply reduced to precisely a set of $(E-1)-4V+10$ independent ones at $k+1$.

\subsection{The `linearized' 4--1 Pachner move}\label{sec_41}

A 4--1 Pachner move pushes an old vertex, labeled by $v^*$, and four old edges adjacent to it, indexed by $o$, into the bulk of the triangulation (see figure \ref{14m}). The momentum updating (\ref{c14b}, \ref{c14c}) of the 4--1 move in linearized form coincides in shape with (\ref{lin32}) of the 3--2 move, with the sole difference that $o$ in this case actually labels four new internal edges and that, accordingly, there are now four pre--constraints. In particular, rewritten as the four linearized equations of motion of the new bulk edges, they read
\ba
\left(N^k_{oo'}+S^\sigma_{oo'}\right)\,y^{o'}_{k}+\left(N^k_{oe}+S^\sigma_{oe}\right)\,y^e_{k}-\Omega^k_{ao}\,y^a_0=0\,,\q\q o=1,\ldots,4\,.\label{a41lineom}
\ea
As shown in appendix \ref{app_41} all coefficients in (\ref{a41lineom}) vanish identically such that the equations of motion of the 4--1 move, in contrast to the one from the 3--2 move, are trivially satisfied. The four pre--constraints of the 4--1 move coincide with the four post--constraints at the four--valent vertex $v^*$. This does not come as a great surprise because, as we already anticipated in section \ref{sec_lincanon}, the vertex displacement generators are constraints which are simultaneously pre-- and post--constraints. As an aside, from the discussion in \cite{Dittrich:2013jaa} it then follows that, thanks to this property, the 4--1 Pachner move of the linearized theory preserves and does {\it not} reduce the rank of the symplectic form (\ref{linomega}) restricted to the post--constraint surface on the evolving slice -- in agreement with the fact that it leaves the number of `gravitons' invariant.

As a result, and in contrast to the 3--2 move, for the 4--1 move the reduction of the transformation matrix $(T_k)$ to the new $(T_{k+1})$ turns out to be trivial. Specifically, this gives $y^e_{k+1}=y^e_k$, $y^\alpha_{k+1}=y^\alpha_k$ and $y^{vI}_{k+1}=y^{vI}_k$. This implies that the embedding of the remaining vertices $v$ does not depend on the edges adjacent to $v^*$ and that the `gravitons' do not depend on the linearized lengths of the removed edges $y^o_{k}$.

Employing the reduction of the transformation matrix, it is further shown in appendix \ref{app_41} that:

\begin{itemize}

\item[(i)] The vertex displacement generators of all surviving $v\subset\Sigma_k\cap\Sigma_{k+1}$ are preserved under the 4--1 Pachner move
\ba
C^k_{vI}=\pi^k_{vI}-N^k_{vI\tilde{w}J}\,y^{\tilde{w}J}_k=0\q\q\q\Rightarrow\q\q\q C^{k+1}_{vI}=\pi^{k+1}_{vI}-N^{k+1}_{vI{w}J}\,y^{{w}J}_{k+1}=0\,,\nn
\ea
where $\tilde{w}$ runs over both $v$ and $v^*$. The `annihilated' `lapse and shift' $y^{v^*J}_{k}$ drop out after the move.

\item[(ii)] The four vertex displacement generators of the removed $v^*$ become trivialized
\ba
C^k_{v^*I}=\pi^k_{v^*I}-N^k_{v^*I\tilde{w}J}\,y^{\tilde{w}J}_k=0\q\q\q\Rightarrow\q\q\q C^{k+1}_{v^*I}=\pi^{k+1}_{v^*I}=0\,.\nn
\ea

\item[(iii)] The `gravitons' experience a {\it `graviton' momentum updating} during the 4--1 move
\ba
\pi^{k+1}_\alpha=\pi^k_\alpha+S^\sigma_{\alpha\beta}\,y^{\beta}_{k}\,.\nn
\ea
(The new `graviton' momenta are determined entirely with the new $(T_{k+1})$.)

\end{itemize}

As in the other moves, the remaining $(y^{vI}_{k+1},\pi^{k+1}_{vI})$ and $(y^\alpha_{k+1},\pi^{k+1}_{\alpha})$ are canonically conjugate pairs of gauge and `graviton' degrees of freedom, respectively; the $(y^\alpha_{k+1},\pi^{k+1}_{\alpha})$ Poisson commute with the surviving $C^{k+1}_{vI}$ and are thus invariant under the vertex displacement gauge symmetry.

Lastly, we note that, since the number of independent `gravitons' is left invariant by the 4--1 move, the six new bulk triangles produced during the move yield six new `gravitons' that are linearly dependent on the already present ones.

%
%
%
%
%

\section{Discussion}\label{sec_disc}

In this article we have given a comprehensive and systematic account of the canonical dynamics of 4D linearized Regge Calculus by means of the elementary and ergodic Pachner evolution moves. The origin of the vertex displacement gauge symmetry of the linearized sector was clarified and the abelian generators of this symmetry were derived. We have identified `lattice gravitons' as gauge invariant and {\it potentially} propagating curvature degrees of freedom. The temporally varying number of linearly independent such `gravitons' on the evolving phase spaces can be systematically counted using the `spatial' triangulation changing Pachner moves. We have elucidated the distinct role of each of the four Pachner evolution moves in the linearized theory and showed that the constraints generating the vertex displacement symmetry are consistent with the dynamics and preserved under all moves. This stands in stark contrast to direct discretizations of continuum constraints which usually result in second class constraints that are not automatically preserved by evolution \cite{Piran:1985ke,Friedman:1986uh,Loll:1997iw}.

One may wonder what happens to the dynamics at first non-linear order in the perturbation around flat Regge solutions. At least for the so-called `tent moves' \cite{Barrett:1994ks,Sorkin:1975ah,Tuckey:1993yf} this has been analyzed in \cite{Dittrich:2009fb}: to second order in the perturbation, the gauge symmetries of Regge Calculus become broken. Consistency conditions arise which can be interpreted as the first (in terms of orders of expansion) equations of motion of the background gauge modes which must propagate once the symmetries get broken. As a result, linearized solutions can generally not be extended to higher order solutions---unless the consistency conditions on the background {\it can} be solved. As shown in \cite{Dittrich:2012qb}, this may generally not be feasible such that perturbative expansions can become inconsistent. This appears analogous to the linearization instabilities in general relativity \cite{fischer,Moncrief:1975un,Moncrief:1976un}. The difference is, however, that in perturbative general relativity consistency conditions arise on the first order physical degrees of freedom, while in perturbative Regge Calculus the consistency conditions constrain the background gauge modes. Furthermore, as shown in \cite{Dittrich:2009fb}, the vertex displacement generators at first non--linear order turn into pseudo constraints \cite{Bahr:2011xs,Piran:1985ke,Friedman:1986uh,Gambini:2002wn,Gambini:2005sv,Gambini:2005vn} with dependence on background data from different time steps.

The abelian Poisson algebra of the vertex displacement generators may be interpreted as the discrete analogue of Dirac's hypersurface deformation algebra in Regge Calculus. These generate 4D symmetry deformations of the hypersurfaces (i.e.\ 4D lattice diffeomorphisms), however, do not generate the dynamics. This is as good as it gets in Regge Calculus because the symmetries become broken to higher order and the generators turn into pseudo constraints such that a hypersurface deformation algebra cannot exist in full 4D Regge Calculus. A consistent hypersurface deformation algebra can only exist in simplicial gravity if the diffeomorphism symmetry is preserved \cite{Bonzom:2013tna}. To this end, one may attempt to change the discretization by coarse graining techniques in order to improve the action order by order such that the symmetry is preserved to higher orders \cite{Dittrich:2014ala,Bahr:2009qc,Bahr:2011uj,Dittrich:2012qb,Dittrich:2011vz,Bahr:2010cq,Dittrich:2013voa}.

The Regge action also emerges in the semiclassical limit of spin foam models for quantum gravity \cite{Perez:2012wv,Conrady:2008mk,Barrett:2009gg} for which, moreover, proposals for a construction of a suitable `graviton propagator' have been made \cite{Rovelli:2005yj,Alesci:2007tg}. The hope is that the results of the present article, through offering a detailed classical understanding of the `lattice graviton' dynamics in linearized Regge Calculus, can likewise contribute to a better understanding of the `graviton' dynamics in spin foam models. The hope is also that the present work may provide novel insights into connecting the covariant spin foam with the canonical loop quantum gravity (LQG) dynamics. For example, a regularization of the LQG Hamiltonian constraint, motivated from a spin foam perspective, has been put forward in \cite{Alesci:2010gb} and was shown to generate 1--4 Pachner moves in the spin network, but the other moves were left open. This article suggests that all four Pachner evolution moves should, indeed, be considered in order to get a non-trivial and complete dynamics -- at least in Regge Calculus, and therefore, presumably, also in spin foams and LQG. As seen, starting from a single simplex, a pure 1--4 move evolution generates stacked spheres with trivial dynamics. The other moves are needed in order to get `gravitons' into the picture.

Finally, we note that a consistent framework for a quantization of linearized Regge Calculus on a flat background already exists. The quantum formalism in \cite{Hoehn:2014wwa} has been designed precisely for a quantization of the Pachner moves that directly yields an equivalence between the canonical and covariant dynamics also in the quantum theory. This quantization framework has thus far been spelled out in detail for variational discrete systems with flat Euclidean configuration spaces $\cq=\mathbb{R}^N$. But this is all that is needed for linearized Regge Calculus: while in the full theory one has $\cq=\mathbb{R}^N_+$ (lengths cannot be negative), in the linearized theory we can, in principle, have $y^e_k\in\mathbb{R}$ (if the expansion parameter $\varepsilon$ is sufficiently small). In particular, a generalized form of local evolution moves has been quantized and their distinct coarse graining and refining roles on evolving Hilbert spaces has been studied in \cite{Hoehn:2014wwa} in close analogy to the present classical investigations. These results are directly applicable to the linearized Pachner evolution moves. To end with a specific example, these results imply that the quantum linearized 4--1 move always produces divergences in a state sum because of the vertex displacement symmetry.

\appendix

\section{Degeneracies of Hessian and Lagrangian two--form}\label{app_deg}

In this appendix we shall demonstrate the statements made in section \ref{sec_deg} according to which the `spatial' components $Y^e_{vI}$ for $v\subset\Sigma_k$ define degeneracies of `effective' Hessians and Lagrangian two-forms on flat Regge triangulations. 

Denote by $S=S_{k}+S_{\text{rest}}$ the action contribution of the flat triangulation depicted in figure \ref{4dstarcomp}, where $S_{\text{rest}}$ denotes the action contribution from $\ct_{\text{rest}}$. We use the notation explained in section \ref{sec_deg} and, furthermore, label any edges adjacent to $v$ in $\ct_{\text{rest}}$ which do {\it not} lie in $\Sigma_k$ by $q$ with $l^q_k$ denoting their corresponding lengths. 

Next, choose a vector $Y_{vI}$, for the specific $v\subset\Sigma_k$, whose 4D star has been completed. (\ref{bia3}, \ref{p2SR}) imply
\ba\label{appYvan}
&&Y^e_{vI}\frac{\partial^2\,S}{\partial l^e_k\partial l^i_k}+Y^{i'}_{vI}\frac{\partial^2\,S}{\partial l^{i'}_k\partial l^i_k}+Y^{q}_{vI}\frac{\partial^2\,S}{\partial l^{q}_k\partial l^i_k}=Y^e_{vI}\frac{\partial^2\,S_{k}}{\partial l^e_k\partial l^i_k}+Y^{i'}_{vI}\frac{\partial^2\,S_{k}}{\partial l^{i'}_k\partial l^i_k}=0\,,\nn\\
&&Y^e_{vI}\frac{\partial^2\,S}{\partial l^e_k\partial l^a_0}+Y^{i'}_{vI}\frac{\partial^2\,S}{\partial l^{i'}_k\partial l^a_0}+Y^{q}_{vI}\frac{\partial^2\,S}{\partial l^{q}_k\partial l^a_0}=Y^e_{vI}\frac{\partial^2\,S_{k}}{\partial l^e_k\partial l^a_0}+Y^{i'}_{vI}\frac{\partial^2\,S_{k}}{\partial l^{i'}_k\partial l^a_0}=0\,.
\ea
Let us explain the first equalities. The last terms on the left hand sides of (\ref{appYvan}) vanish since the edges labeled by $q$ will not share any 4--simplex with any of the edges labeled by $a,i$ and, hence, second derivatives of $S$ with respect to a pair of length variables associated to such a pair of edges must vanish. Additionally, since the only simplices which contain pairs of edges from the set labeled by $a,i$ or pairs of edges from both the set labeled by $a,i$ and $e$ already occur in the triangulation at step $k$, we can restrict the second partial derivatives of $S$ in the remaining terms to the second partial derivatives of $S_{k}$ and the expressions on the right hand sides of (\ref{appYvan}) are obtained. 

Using (\ref{lagrange}), the second line directly implies (see also appendix A of \cite{Dittrich:2013jaa} for $\Omega^k$)
\ba
\Omega^{k}_{ae}Y^e_{vI}+\Omega^k_{ai}Y^i_{vI}=-Y^e_{vI}\frac{\partial^2\,S_{k}}{\partial l^e_k\partial l^a_0}-Y^{i'}_{vI}\frac{\partial^2\,S_{k}}{\partial l^{i'}_k\partial l^a_0}=0\,.
\ea

Let us now show that the `spatial' components of $Y^e_{vI}$ associated to the edges $e\subset\Sigma_k$ define also right null vectors of the `effective' Lagrangian two--form $\tilde{\Omega}^{k}$, corresponding to the `effective' action $\tilde{S}_{k}$ with the bulk lengths $l^i_k$ integrated out.
It is straightforward to check (e.g., see appendix A of \cite{Dittrich:2013jaa}) that the `effective' Lagrangian two--form (\ref{lagrange2}) reads
\ba\label{appeff-omega2}
\frac{\partial^2\,\tilde{S}_{k}}{\partial l^e_k\partial l^a_0}=\frac{\partial^2\,S_{k}}{\partial l^e_k\partial l^a_0}-\frac{\partial^2\,S_{k}}{\partial l^e_k\partial l^i_k}\,T^i_A\left(T^{i_1}_A\frac{\partial^2\,S_{k}}{\partial l^{i_1}_k\partial l^{i_2}_k}T^{i_2}_{A'}\right)^{-1}T^{i'}_{A'}\,\frac{\partial^2\,S_{k}}{\partial l^{i'}_k\partial l^a_0}\,,
\ea
where $T^i_A$ are linearly independent non--degenerate directions of the Hessian $\frac{\partial^2\,S_{k}}{\partial l^i_k\partial l^{i'}_k}$ of the bulk between $\Sigma_0$ and $\Sigma_k$. We need to project the latter matrix with these $T^i_A$ in order to factor out the degenerate directions and render the resulting matrix invertible (see also \cite{Hoehn:2014aoa} for a related discussion). We note that $Y^i_{vI}$ for $v\subset\Sigma_k$ will generally {\it not} define a null vector of $\frac{\partial^2\,S_{k}}{\partial l^i_k\partial l^{i'}_k}$ because the degenerate directions of the latter will correspond to displacements in flat directions of vertices in the bulk of the triangulation between $\Sigma_0$ and $\Sigma_k$ and not in the boundary surface $\Sigma_k$. Otherwise, the first equation in (\ref{appYvan}) would imply $Y^e_{vI}\frac{\partial^2\,S_{k}}{\partial l^e_k\partial l^i_k}=0$ which is generally not possible. Accordingly, we may choose the four vectors $Y^i_{vI}$ to be contained in the set $T^i_A$.

Using the right hand sides of both equations in (\ref{appYvan}), one finds 
\ba\label{appomegadeg}
\tilde{\Omega}^{k}_{ae}Y^e_{vI}&=&-Y^e_{vI}\frac{\partial^2\,\tilde{S}_{k}}{\partial l^e_k\partial l^a_0}=-Y^e_{vI}\frac{\partial^2\,S_{k}}{\partial l^e_k\partial l^a_0}-\underset{=\,\delta_{vI}^{A'}}{\underbrace{Y^{i'}_{vI}\,\f{\p^2\, S_{k}}{\p l^{i'}_k\p l^i_k}\,T^i_A\left(T^{i_1}_A\frac{\partial^2\,S_{k}}{\partial l^{i_1}_k\partial l^{i_2}_k}T^{i_2}_{A'}\right)^{-1}}}T^{j}_{A'}\,\frac{\partial^2\,S_{k}}{\partial l^{j}_k\partial l^a_0}\nn\\
&=&-Y^e_{vI}\frac{\partial^2\,S_{k}}{\partial l^e_k\partial l^a_0}-Y^{i'}_{vI}\frac{\partial^2\,S_{k}}{\partial l^{i'}_k\partial l^a_0}=0\,.
\ea

Finally, we shall briefly demonstrate that the $Y^e_{vI}$ constitute degenerate directions of the `effective Hessian' with edges labeled by both $i$ and $q$ integrated out. Namely, consider the completed 4D star of the vertex $v\subset\Sigma_k$ as given above with $\ct_{\text{rest}}$ glued onto $\Sigma_k$. Given that $Y_{vI}$ is a null vector of the (non--effective) Hessian and in analogy to (\ref{appYvan}), we must have 
\ba\label{appYvan2}
Y^e_{vI}\frac{\partial^2\,S}{\partial l^e_k\partial l^i_k}+Y^{i'}_{vI}\frac{\partial^2\,S}{\partial l^{i'}_k\partial l^i_k}&=&0\,,\nn\\
Y^e_{vI}\frac{\partial^2\,S}{\partial l^e_k\partial l^{e'}_k}\,\,\,+Y^{i}_{vI}\frac{\partial^2\,S}{\partial l_k^{i}\partial l_k^{e'}}&=&0\,,
\ea
where for notational simplicity we have here combined the two indices $i$ and $q$ into the single index $i$. Analogously to (\ref{appeff-omega2}), the `effective' Hessian of the effective action $\tilde{S}$ with $l^i_k,l^q_k$ integrated out reads
\ba\label{eff-hesse}
\tilde{H}_{ee'}:=\frac{\partial^2\,\tilde{S}}{\partial l^e_k\partial l^{e'}_k}=\frac{\partial^2\,S}{\partial l^e_k\partial l_k^{e'}}-\frac{\partial^2\,S}{\partial l_k^e\partial l_k^i}\,T^i_A\left(T^{i_1}_A\frac{\partial^2\,S}{\partial l_k^{i_1}\partial l_k^{i_2}}T^{i_2}_{A'}\right)^{-1}T^{i'}_{A'}\,\frac{\partial^2\,S}{\partial l_k^{i'}\partial l_k^{e'}}\,.
\ea
In conjunction with (\ref{appYvan2}), and in analogy to (\ref{appomegadeg}), one finds the desired result
\ba
Y^e_{vI}\tilde{H}_{ee'}=0\,.\nn
\ea

\section{Linearized Pachner moves}\label{app_linpach}

In this appendix we shall confirm the claims of section \ref{sec_linpach} concerning the linearized canonical Pachner move dynamics. As regards notation: sometimes we shall use an index $c$ to label both $e,n$ or $e,o$.

\subsection{The `linearized' 1--4 Pachner move}\label{app_14}

We shall now demonstrate the statements of section \ref{sec_14}.

Consider a hypersurface $\Sigma_k$ and assume the transformation matrix $(T_k)^e_\Gamma$ has been chosen according to the prescription in section \ref{sec_tmatrix}. That is, at step $k$ we have
\ba\label{14oldsplit}
y^e_k=(T_k)^e_{vI}y^{vI}_k+(T_k)^e_\alpha y^\alpha_k\,,\q\q\q y^{vI}_k=(T_k^{-1})^{vI}_ey^e_k\,,\q\q
\q y^\alpha_k=(T_k^{-1})^\alpha_ey^e_k\,.
\ea
Perform a 1--4 Pachner move on $\Sigma_k$ yielding a new vertex $v^*$ (see figure \ref{14m}). Since we now have four new edges labeled by $n$ and four new gauge variables enumerated by $v^*I$ (the 1--4 move does not introduce new bulk triangles) we must extend the transformation matrix at step $k+1$ suitably. This extended matrix must be in agreement with the prescription in section \ref{sec_tmatrix} and should yield the new decomposition
\ba\label{14newsplit}
y^e_{k+1}&=&(T_{k+1})^e_{vI}y^{vI}_{k+1}+(T_{k+1})^e_\alpha y^\alpha_{k+1}+(T_{k+1})^e_{v^*I}y^{v^*I}_{k+1}\,,\nn\\
 y^n_{k+1}&=&(T_{k+1})^n_{vI}y^{vI}_{k+1}+(T_{k+1})^n_\alpha y^\alpha_{k+1}+(T_{k+1})^n_{v^*I}y^{v^*I}_{k+1}\,,\nn\\
y^{vI}_{k+1}&=&(T_{k+1}^{-1})^{vI}_ey^e_{k+1}+(T_{k+1}^{-1})^{vI}_ny^n_{k+1}\,,\nn\\
y^\alpha_{k+1}&=&(T_{k+1}^{-1})^\alpha_ey^e_{k+1}+(T_{k+1}^{-1})^\alpha_ny^n_{k+1}\,,\nn\\
 y^{v^*I}_{k+1}&=&(T_{k+1}^{-1})^{v^*I}_ey^e_{k+1}+(T_{k+1}^{-1})^{v^*I}_ny^n_{k+1}\,.
\ea

In order to maintain the same linearized configuration coordinates for the old vertices, $y^{vI}_{k+1}=y^{vI}_k$ ($y^e_{k+1}=y^e_k$ and $y^n_{k+1}$ are not needed in order to determine the embedding of the old vertices), we set
\ba\label{14Toldchoice}
(T_{k+1})^e_\Gamma=(T_k)^e_\Gamma\,,\q\q (T_{k+1}^{-1})^\Gamma_e=(T_k^{-1})^\Gamma_e\,,
\ea
where $\Gamma$ runs over the old $vI$ and $\alpha$ (but does not include the $v^*I$). Include both $\Gamma$ and the four $v^*I$ in a new index $\Lambda$, and $e$ and $n$ in the index $c$. Using
\ba
(T_{k+1})^c_\Lambda(T_{k+1}^{-1})^\Lambda_{c'}=\delta^c_{c'}\,,\q\q\q\q (T_{k+1}^{-1})^\Lambda_{c}(T_{k+1})^c_{\Lambda'}=\delta^\Lambda_{\Lambda'}\,,\nn
\ea
it is straightforward to convince oneself that the new components of the transformation matrix at $k+1$ can accordingly be chosen as (note that $y^n_{k+1}$ do not contribute to any `gravitons')
\ba\label{14Tsol}
(T_{k+1})^n_{vI}&=&(Y_{k+1})^n_{vI}\,,\q\q (T_{k+1})^n_{v^*I}=(Y_{k+1})^n_{v^*I}=\delta^n_I\,,\nn\\
(T_{k+1})^n_\alpha&=&0\,,\q\q\q\q\q\,\, (T_{k+1})^e_{v^*I}=(Y_{k+1})^e_{v^*I}=0\,,
\ea
with inverse\footnote{The precise form of $(T_{k+1}^{-1})^{v^*I}_e$ is not relevant for us. Notice, however, that it cannot vanish, since, apart from the $y^n_{k+1}$, some of the $y^e_{k+1}$ are necessary in order to specify the embedding of the new vertex (the position of a vertex $v$ also depends on edges in the boundary of $\text{star}_{4D}(v)$). }
\ba
(T_{k+1}^{-1})^{vI}_n&=&0\,,\q\q \,(T_{k+1}^{-1})^{v^*I}_n=\delta^I_n\,,\nn\\
(T_{k+1}^{-1})^{v^*I}_e&\neq& 0\,,\q\q\q(T_{k+1})^\alpha_n=0\,.\nn
\ea

Given this choice of the new matrix $(T_{k+1})^c_\Lambda$, we may study the behaviour of the gauge variables and `gravitons' under the linearized momentum updating equations (\ref{lin14}) of the 1--4 move. 
Let us begin by considering the momenta conjugate to the old gauge degrees of freedom. Using (\ref{14Toldchoice}, \ref{14Tsol}, \ref{lin14}), we find (recall that $c$ runs over both $e$ and $n$)
\ba\label{14Ymomup}
(Y_{k+1})^c_{vI}\pi^{k+1}_c\hspace{-.15cm}&=&(Y_{k+1})^c_{vI}\left(\pi^k_c+S^\sigma_{cc'}y^{c'}_{k+1}\right)\nn\\
&\underset{\pi^k_n=0}{=}&\pi^k_{vI}+N^{k}_{vI\alpha}y^{\alpha}_{k}+S^\sigma_{vI\alpha}y^\alpha_{k+1}+S^\sigma_{vIwJ}y^{wJ}_{k+1}+S^\sigma_{vI v^*J}y^{v^*J}_{k+1}\,,
\ea
where in the last equation we have made use of (\ref{ncantrans}) and (\ref{14newsplit}). As a consequence of the absence of equations of motion for the 1--4 move, one finds in the present case $N^{k+1}_{cc'}=N^{k}_{cc'}+S^\sigma_{cc'}$ where $N^k_{ee'}$ is defined in (\ref{matrixdefs}).\footnote{Notice that $N^k_{en}=0$.} As in (\ref{ncantrans}), the new momenta conjugate to the gauge modes $y^{vI}_{k+1}$ are then (recall $y^\alpha_{k+1}=y^\alpha_k$)
\ba\label{14momgaup}
\pi^{k+1}_{vI}&:=&(Y_{k+1})^c_{vI}\pi^{k+1}_c-N^{k+1}_{vI\alpha}y^{\alpha}_{k+1}\nn\\
&=&\pi^k_{vI}+S^\sigma_{vI\tilde{v}J}y^{\tilde{v}J}_{k+1}\,,
\ea
where $\tilde{v}$ now includes both $v$ and $v^*$. Solving (\ref{14purecons}) for $\pi^k_{vI}$, inserting this into (\ref{14momgaup}) and noting that $y^{vI}_k=y^{vI}_{k+1}$, the previous apparent `evolution equations' rather transform into the new constraints at $k+1$ generating the vertex displacement of $v$ in $\Sigma_{k+1}$,
\ba\label{14newcon2}
C^{k+1}_{vI}=\pi^{k+1}_{vI}-N^{k+1}_{vI\tilde{v}J}y^{\tilde{v}J}_{k+1}\,,
\ea
which are, thus, preserved. Proceeding similarly with the new gauge modes $y^{v^*I}_{k+1}$, one finds the four new constraints introduced by the 1--4 move which generate the displacement of the new vertex $v^*$ in $\Sigma_{k+1}$ as
\ba\label{14momgaup2}
\pi^{k+1}_{v^*I}&:=&(Y_{k+1})^c_{v^*I}\pi^{k+1}_c-S^\sigma_{v^*I\alpha}y^{\alpha}_{k+1}\nn\\
&=& S^\sigma_{v^*I\tilde{v}J}y^{\tilde{v}J}_{k+1}\,,
\ea
where (\ref{14Tsol}) and $\pi^k_{v^*I}=(Y_{k+1})^n_{v^*I}\pi^k_n=0$ was used (recall that $\pi^k_n=0$).

Finally, let us examine the evolution of the `graviton' momenta. In analogy to (\ref{14Ymomup}),
\ba
(T_{k+1})^c_\alpha\pi^{k+1}_c&=&(T_{k+1})^c_\alpha\left(\pi^k_c+S^\sigma_{cc'}y^{c'}_{k+1}\right)\nn\\
&=&\pi^k_\alpha+N^{k}_{\alpha vI}y^{vI}_{k}+S^\sigma_{\alpha\beta}y^\beta_{k+1}+S^\sigma_{\alpha vI}y^{vI}_{k+1}+S^\sigma_{\alpha v^*I}y^{v^*I}_{k+1}\,,\nn
\ea
such that, using (\ref{ncantrans}) at $k+1$ and noting that $y^\alpha_k=y^\alpha_{k+1}$,
\ba\label{14momobup}
\pi^{k+1}_\alpha&:=&(T_{k+1})^c_\alpha\pi^{k+1}_c-N^{k+1}_{\alpha\tilde{v}I}y^{\tilde{v}I}_{k+1}\nn\\
&=&\pi^k_\alpha+S^\sigma_{\alpha\beta}y^\beta_{k}\,.
\ea
In contrast to (\ref{14momgaup}), these are generally {\it not} constraints. 


\subsection{The `linearized' 2--3 Pachner move}\label{app_23}

Next, we establish the statements of section \ref{sec_23}.

As before, take a hypersurface $\Sigma_k$ and assume $(T_k)^e_\Gamma$ has been chosen in accordance with the prescription of section \ref{sec_tmatrix} such that (\ref{14oldsplit}) holds. Now perform a 2--3 Pachner move on $\Sigma_k$ which introduces one new edge labeled by $n$ and a new {\it a priori} free deficit angle $\epsilon^{\alpha^*}$ (see figure \ref{23m}). The transformation matrix must be extended in a suitable way in order to incorporate the new degrees of freedom in the splitting between `graviton' and gauge variables, such that after the 2--3 move
\ba\label{23newsplit}
y^e_{k+1}&=&(T_{k+1})^e_{vI}y^{vI}_{k+1}+(T_{k+1})^e_\alpha y^\alpha_{k+1}+(T_{k+1})^e_{\alpha^*}y^{\alpha^*}_{k+1}\,,\nn\\
 y^n_{k+1}&=&(T_{k+1})^n_{vI}y^{vI}_{k+1}+(T_{k+1})^n_\alpha y^\alpha_{k+1}+(T_{k+1})^n_{\alpha^*}y^{\alpha^*}_{k+1}\,,\nn\\
y^{vI}_{k+1}&=&(T_{k+1}^{-1})^{vI}_ey^e_{k+1}+(T_{k+1}^{-1})^{vI}_ny^n_{k+1}\,,\nn\\
\q y^\alpha_{k+1}&=&(T_{k+1}^{-1})^\alpha_ey^e_{k+1}+(T_{k+1}^{-1})^\alpha_ny^n_{k+1}\,,\nn\\
 y^{\alpha^*}_{k+1}&=&(T_{k+1}^{-1})^{\alpha^*}_ey^e_{k+1}+(T_{k+1}^{-1})^{\alpha^*}_ny^n_{k+1}\,.
\ea
Again, $y^e_{k+1}=y^e_k$, $y^\alpha_{k+1}=y^\alpha_k$ and we also choose to keep $y^{vI}_{k+1}=y^{vI}_k$. The extension of $(T_k)^e_\Gamma$ at step $k+1$ can therefore be performed in complete analogy to the extension of the transformation matrix in the course of the 1--4 move in appendix \ref{app_14}---just replacing $v^*I$ by $\alpha^*$ in the equations and noting that $n$ now labels a single edge. In particular, we again keep (\ref{14Toldchoice}) and, in analogy to (\ref{14Tsol}), find
\ba\label{23Tsol}
(T_{k+1})^n_{vI}&=&(Y_{k+1})^n_{vI}\,,\q\q(T_{k+1})^n_{\alpha^*}=\frac{1}{(T_{k+1}^{-1})^{\alpha^*}_n}\,,\nn\\
(T_{k+1})^e_{\alpha^*}&=&0\,,\q\q\q\q\,\,\,\,\q(T_{k+1})^n_\alpha=-\frac{(T_{k+1}^{-1})^{\alpha^*}_e}{(T_{k+1}^{-1})^{\alpha^*}_n}(T_{k+1})^e_\alpha\neq0\,,
\ea
with inverse ($y^n_{k+1}$ does not contribute to the deficit angles inherited from step $k$)
\ba
(T_{k+1}^{-1})^{\alpha^*}_n&=&\frac{\partial\,\tilde{\epsilon}^{\alpha^*}}{\partial l^n_{k+1}}\,,\q\q(T_{k+1}^{-1})^{vI}_n=0\,,\nn\\
(T_{k+1}^{-1})^{\alpha^*}_e&=&\frac{\partial\,\tilde{\epsilon}^{\alpha^*}}{\partial l^e_{k+1}}\,,\q\q(T_{k+1}^{-1})^\alpha_n\,\,=0\,.\nn
\ea
We emphasise that, because the length $l^n_{k+1}$ of the new edge introduced in the 2--3 move determines the new deficit angle, we generically have $\frac{\partial\,\tilde{\epsilon}^{\alpha^*}}{\partial l^n_{k+1}}\neq0$ and so the components on the right hand side of (\ref{23Tsol}) are well defined. At this stage, the new matrix of step $k+1$ is chosen in agreement with section \ref{sec_tmatrix}. It should be noted that the only non--vanishing component of the vector $(T_{k+1})_{\alpha^*}$ is $(T_{k+1})^n_{\alpha^*}$ corresponding to the new {\it a priori} free edge of the 2--3 move. As can be easily checked, this vector is therefore a right null vector at step $k+1$, i.e.\ $\Omega^{k+1}_{ac}(T_{k+1})^c_{\alpha^*}=0$, where $c$ labels both $e,n$, in agreement with the fact that the new `graviton' $y^{\alpha^*}_{k+1}$ is an {\it a priori} free variable.

Using the new transformation matrix, let us now study the time evolution equations. The momentum updating of the 2--3 move in linearized form is in shape identical to (\ref{lin14}) except that $n$ now labels a single new edge. For the momenta conjugate to the gauge modes we find, in analogy to (\ref{14Ymomup}),
\ba
(Y_{k+1})^c_{vI}\pi^{k+1}_c&=&(Y_{k+1})^c_{vI}\left(\pi^k_c+S^\sigma_{cc'}y^{c'}_{k+1}\right)\nn\\
&=&\pi^k_{vI}+N^{k}_{vI\alpha}y^{\alpha}_{k}+S^\sigma_{vI\alpha}y^\alpha_{k+1}+S^\sigma_{vIwJ}y^{wJ}_{k+1}+S^\sigma_{vI \alpha^*}y^{\alpha^*}_{k+1}\,,\nn
\ea
such that (again, $N^{k+1}_{cc'}=N^k_{cc'}+S^\sigma_{cc'}$ because there is no new equation of motion)
\ba\label{23momgaup}
\pi^{k+1}_{vI}&:=&(Y_{k+1})^c_{vI}\pi^{k+1}_c-N^{k+1}_{vI\tilde{\alpha}}y^{\tilde{\alpha}}_{k+1}\nn\\
&=&\pi^k_{vI}+S^\sigma_{vIwJ}y^{wJ}_{k+1}\,,
\ea
where $\tilde{\alpha}$ runs over both $\alpha$ and $\alpha^*$. Solving (\ref{14purecons}) for $\pi^k_{vI}$, (\ref{23momgaup}), again, transforms into the new constraints, generating the displacement of vertices of $\Sigma_{k+1}$ in flat directions,
\ba\label{23newcon4}
C^{k+1}_{vI}=\pi^{k+1}_{vI}-N^{k+1}_{vIwJ}y^{wJ}_{k+1}\,,
\ea
which are thus preserved under the 2--3 move.

Likewise, for the momenta conjugate to the old `gravitons' one finds
\ba
(T_{k+1})^c_\alpha\pi^{k+1}_c&=&(T_{k+1})^c_\alpha\left(\pi^k_c+S^\sigma_{cc'}y^{c'}_{k+1}\right)\nn\\
&=&\pi^k_\alpha+N^{k}_{\alpha vI}y^{vI}_{k}+S^\sigma_{\alpha\beta}y^\beta_{k+1}+S^\sigma_{\alpha vI}y^{vI}_{k+1}+S^\sigma_{\alpha \alpha^*}y^{\alpha^*}_{k+1}\,,\nn
\ea
which yields as in (\ref{14momobup}),
\ba\label{23momobup}
\pi^{k+1}_\alpha&:=&(T_{k+1})^c_\alpha\pi^{k+1}_c-N^{k+1}_{\alpha vI}y^{vI}_{k+1}\nn\\
&=&\pi^k_\alpha+S^\sigma_{\alpha\tilde{\beta}}y^{\tilde{\beta}}_{k}\,.\nn
\ea
Similarly, using (\ref{23Tsol}), $\pi^k_n=0$ and the fact that the `old gravitons' satisfy $y^{{\beta}}_{k}=y^{{\beta}}_{k+1}$, the momentum conjugate to the newly generated `graviton' $y^{\alpha^*}_{k+1}$ reads
\ba
\pi^{k+1}_{\alpha^*}&:=&(T_{k+1})^n_{\alpha^*}\pi^{k+1}_n-S^\sigma_{\alpha^* vI}y^{vI}_{k+1}=(T_{k+1})^n_{\alpha^*}\pi^k_n+S^\sigma_{\alpha^* \tilde{\beta}}y^{\tilde{\beta}}_{k+1}=\pi^k_{\alpha^*}+S^\sigma_{\alpha^* \tilde{\beta}}y^{\tilde{\beta}}_{k+1}\nn\\
&=&S^\sigma_{\alpha^* \tilde{\beta}}y^{\tilde{\beta}}_{k+1}\,.\nn
\ea

\subsection{The `linearized' 3--2 Pachner move}\label{app_32}

Here we shall support the claims made in section \ref{sec_32}.

Consider a hypersurface $\Sigma_k$ on which we shall perform a 3--2 Pachner move which renders an old edge labeled by $o$ internal (see figure \ref{23m}). Assume the transformation matrix $(T_k)^c_\Gamma$, where $c$ runs over both $e$ and $o$ and $\Gamma$ runs over $vI,\alpha,\alpha^*$, is chosen according to the prescription of section \ref{sec_tmatrix} such that
\ba\label{32oldsplit}
y^e_{k}&=&(T_k)^e_{vI}y^{vI}_{k}+(T_k)^e_\alpha y^\alpha_{k}+(T_k)^e_{\alpha^*}y^{\alpha^*}_{k}\,,\nn\\
 y^o_{k}&=&(T_k)^o_{vI}y^{vI}_{k}+(T_k)^o_\alpha y^\alpha_{k}+(T_k)^o_{\alpha^*}y^{\alpha^*}_{k}\,,\nn\\
y^{vI}_{k}&=&(T_k^{-1})^{vI}_ey^e_{k}+(T_k^{-1})^{vI}_oy^o_{k}\,,\nn\\
 y^\alpha_{k}&=&(T_k^{-1})^\alpha_ey^e_{k}+(T_k^{-1})^\alpha_oy^o_{k}\,,\nn\\
 y^{\alpha^*}_{k}&=&(T_k^{-1})^{\alpha^*}_ey^e_{k}+(T_k^{-1})^{\alpha^*}_oy^o_{k}\,,
\ea
where $y^{\alpha^*}_k$ is such that the old edge has a non--vanishing contribution, i.e.\ $(T^{-1})^{\alpha^*}_o\neq 0$ (such $y^{\alpha^*}_k$ generically exists). This will be the `annihilated graviton'.

Let $E$ be the number of edges in $\Sigma_k$. The task is to appropriately reduce the $E\times E$ transformation matrix of step $k$ to a new $(E-1)\times(E-1)$ matrix at step $k+1$ which likewise disentangles the $4V-10$ gauges modes $y^{vI}_{k+1}$ from the $(E-1)-4V+10$ `gravitons' $y^\alpha_{k+1}$ in $\Sigma_{k+1}$ and agrees with the prescription of section \ref{sec_tmatrix}. After the move it should yield
\ba\label{32newsplit}
y^e_{k+1}&=&(T_{k+1})^e_{vI}y^{vI}_{k+1}+(T_{k+1})^e_\alpha y^\alpha_{k+1}\,,\q\q y^{vI}_{k+1}=(T_{k+1}^{-1})^{vI}_ey^e_{k+1}\,\nn\\ y^\alpha_{k+1}&=&(T_{k+1}^{-1})^\alpha_ey^e_{k+1}\,.
\ea
To this end, we must make use of the equation of motion (\ref{32geom}) or, equivalently, the pre--constraint of the 3--2 move. 
Thanks to the results of section \ref{sec_deg}, one may convince oneself that
\ba\label{32Yhesse}
(Y_{k})^c_{vI}\left(N^k_{co}+S^\sigma_{co}\right)=0\,.
\ea
In conjunction with the decomposition (\ref{32oldsplit}), this implies that (\ref{32geom}) can be written as
\ba\label{32gravlindep}
\left(N^k_{o\alpha^*}+S^\sigma_{o\alpha^*}\right)y^{\alpha^*}_k+\left(N^k_{o\alpha}+S^\sigma_{o\alpha}\right)y^{\alpha}_k+\Omega^k_{\gamma o}y^\gamma_0=0
\ea
(use of $\Omega^k_{ao}(Y_0)^a_{vI}=0$ and a similar decomposition for $\Sigma_0$ has been made).

We can employ this equation to produce the new $(T_{k+1})$ from $(T_k)$. This will require some work. Firstly, we choose $y^{\alpha^*}_k$ to be the `graviton' that either gets `annihilated' or fixed by (\ref{32gravlindep}) (it was chosen to depend on $y^o_k$). For this `graviton' we may keep the old decomposition and set
\ba
(T_{k+1}^{-1})^{\alpha^*}_c&:=&(T_{k}^{-1})^{\alpha^*}_c\,,\nn
\ea
such that $y^{\alpha^*}_k=y^{\alpha^*}_{k+1}$.

Next, solve the pre--constraint in the form (\ref{32geom}) for $y^o_k(y^e_k,y^a_0)$ (generically, $N^k_{oo}+S^\sigma_{oo}\neq0$) and insert the solution into (\ref{32oldsplit}), in order to rewrite the expressions for $y^{vI}_k,y^\alpha_k$. It gives
\ba\label{32split1}
y^{vI}_k&=&(T_{k+1}^{-1})^{vI}_ey^e_{k}+(T_{k+1}^{-1})^{vI}_ay^a_0=:y^{vI}_{k+1}+\delta y^{vI}_0\,,\nn\\
y^{\alpha}_k&=&(T_{k+1}^{-1})^\alpha_ey^e_{k}\,\,+\,(T_{k+1}^{-1})^\alpha_ay^a_0\,=:y^{\alpha}_{k+1}+\delta y^\alpha_0\,,
\ea
where it can be easily checked that the coefficients of the new (effective) inverse transformation matrix (with $y^o_k$ integrated out) read
\ba
(T_{k+1}^{-1})^{vI}_e&:=&(T_{k}^{-1})^{vI}_e-(T_{k}^{-1})^{vI}_o\left(N^k_{oo}+S^\sigma_{oo}\right)^{-1}\left(N^k_{oe}+S^\sigma_{oe}\right)\,,\nn\\
(T_{k+1}^{-1})^{vI}_o&:=&0\,,\nn\\
(T_{k+1}^{-1})^{vI}_a&:=&(T_{k}^{-1})^{vI}_o\left(N^k_{oo}+S^\sigma_{oo}\right)^{-1}\Omega^k_{ao}\,,\nn\\
(T_{k+1}^{-1})^\alpha_e&:=&(T_{k}^{-1})^\alpha_e-(T_{k}^{-1})^\alpha_o\left(N^k_{oo}+S^\sigma_{oo}\right)^{-1}\left(N^k_{oe}+S^\sigma_{oe}\right)\,,\nn\\
(T_{k+1}^{-1})^\alpha_o&:=&0\,,\nn\\
(T_{k+1}^{-1})^\alpha_a&:=&(T_{k}^{-1})^\alpha_o\left(N^k_{oo}+S^\sigma_{oo}\right)^{-1}\Omega^k_{ao}\,.\nn
\ea
Hence, using that by (\ref{lin32}) $y^e_{k+1}=y^e_k$ and dropping the terms $\delta y^{vI}_0,\delta y^\alpha_0$ depending on the initial data,
\ba
y^{vI}_{k+1}=(T_{k+1}^{-1})^{vI}_ey^e_{k+1}\neq y^{vI}_k \,,\q\q  y^\alpha_{k+1}=(T_{k+1}^{-1})^\alpha_ey^e_{k+1}\neq y^\alpha_k\,.\nn
\ea 
From step $k+1$ onwards we will employ these shifted $y^\alpha_{k+1}$ and $y^{vI}_{k+1}$ as `graviton' and gauge variables, respectively. 


We proceed by using (\ref{32gravlindep}) to solve for $y^{\alpha^*}_k$ as a function of $y^\alpha_k$ and $y^\gamma_0$, and rewrite the first equation in (\ref{32oldsplit}),
\ba\label{32ye}
y^e_{k}=(T_{k+1})^e_\alpha y^\alpha_k+(T_{k+1})^e_{vI}y^{vI}_k+(T_{k+1})^e_\gamma y^\gamma_0
\ea
which gives the components of the new (effective) $(T_{k+1})$ as follows
\ba
(T_{k+1})^e_\alpha&:=&(T_k)^e_\alpha-(T_k)^e_{\alpha^*}\left(N^k_{o\alpha^*}+S^\sigma_{o\alpha^*}\right)^{-1}\left(N^k_{o\alpha}+S^\sigma_{o\alpha}\right)\,,\nn\\
(T_{k+1})^e_{\alpha^*}&:=&0\,,\nn\\
(T_{k+1})^e_{vI}&:=&(Y_k)^e_{vI}\,,\nn\\
(T_{k+1})^e_\gamma&:=&(T_k)^e_{\alpha^*}\left(N^k_{o\alpha^*}+S^\sigma_{o\alpha^*}\right)^{-1}\Omega^k_{\gamma o}\,.\nn
\ea
Further using the new splitting (\ref{32split1}) and noting that by (\ref{lin32}) $y^e_{k+1}=y^e_k$, (\ref{32ye}) may be conveniently written solely in terms of the new `graviton' and gauge modes (one may convince oneself that the contributions from the initial data drop out)
\ba
y^e_{k+1}=(T_{k+1})^e_\alpha y^\alpha_{k+1}+(Y_{k+1})^e_{vI}y^{vI}_{k+1}\,.\nn
\ea

Finally, making the ansatz
\ba
y^o_{k+1}=(T_{k+1})^o_\alpha y^\alpha_{k+1}+(T_{k+1})^o_{\alpha^*}y^{\alpha^*}_{k+1}+(T_{k+1})^o_{vI}y^{vI}_{k+1}\,,\nn
\ea
one finds that
\ba
(T_{k+1})^o_\alpha&=&(T_k)^o_\alpha-\left((T_k)^o_{\alpha^*}-\frac{1}{(T_k^{-1})^{\alpha^*}_o}\right)\left(N^k_{o\alpha^*}+S^\sigma_{o\alpha^*}\right)^{-1}\left(N^k_{o\alpha}+S^\sigma_{o\alpha}\right)\,,\nn\\
(T_{k+1})^o_{\alpha^*}&=&\frac{1}{(T_k^{-1})^{\alpha^*}_o}\,,\nn\\
(T_{k+1})^o_{vI}&=&(Y_k)^o_{vI}\,,\nn
\ea
yields the remaining components of the new (effective) $(T_{k+1})$ which provides the new decomposition (\ref{32newsplit}), as desired. It is straightforward to check that the new transformation matrix follows the prescription of section \ref{sec_tmatrix} and is an invertible matrix that defines a canonical transformation (provided the old one did). In fact, the new matrix is now in shape analogous to the extended transformation matrix of the 2--3 Pachner move (\ref{23Tsol}) (with $n$ replaced by $o$).

With the transformed matrix in hand, we are in a position to determine the momenta conjugate to the new gauge and `graviton' variables via (\ref{ncantrans}). Noting that 
\ba
N^{k+1}_{ee'}=\left(N^k_{ee'}+S^\sigma_{ee'}\right)-\left(N^k_{eo}+S^\sigma_{eo}\right)\left(N^k_{oo}+S^\sigma_{oo}\right)^{-1}\left(N^k_{oe'}+S^\sigma_{oe'}\right)\,,\nn
\ea
(\ref{32Yhesse}) implies
\ba
(Y_{k+1})^e_{vI}N^{k+1}_{ee'}=(Y_k)^c_{vI}\left(N^k_{ce'}+S^\sigma_{ce'}\right)\,,\label{32YN}
\ea
which allows us to define (recall $(T_{k+1})^e_{\alpha^*}=0$)
\ba
N^{k+1}_{vI\alpha}&:=&(Y_{k+1})^e_{vI}N^{k+1}_{ee'}(T_{k+1})^{e'}_\alpha=(Y_{k+1})^c_{vI}\left(N^k_{ce'}+S^\sigma_{ce'}\right)(T_{k+1})^{e'}_\alpha\,,\nn\\
 N^{k+1}_{vI\alpha^*}&:=&(Y_{k+1})^e_{vI}N^{k+1}_{ee'}({T}_{k+1})^{e'}_{\alpha^*}=(Y_{k+1})^c_{vI}\left(N^k_{ce'}+S^\sigma_{ce'}\right)({T}_{k+1})^{e'}_{\alpha^*}=0\,,\nn\\
 N^{k+1}_{vIwJ}&:=&(Y_{k+1})^e_{vI}N^{k+1}_{ee'}(Y_{k+1})^{e'}_{wJ}=(Y_{k+1})^c_{vI}\left(N^k_{ce'}+S^\sigma_{ce'}\right)(Y_{k+1})^{e'}_{wJ}\, .\nn
\ea
As a result of $\pi^{k+1}_o=0$, this leads via (\ref{ncantrans}) to the new momenta at step $k+1$
\ba
\pi^{k+1}_{vI}&:=&(Y_{k+1})^e_{vI}\pi^{k+1}_e-N^{k+1}_{vI\alpha}y^\alpha_{k+1}\,,\nn\\
\pi^{k+1}_\alpha&:=&(T_{k+1})^e_\alpha\pi^{k+1}_e-N^{k+1}_{\alpha vI}y^{vI}_{k+1}\,,\nn
\ea
(Both new sets of momenta are computed entirely from variables and matrix components associated to $\Sigma_{k+1}$.) It is not difficult to verify that the shifted variables $(y^{vI}_{k+1},\pi^{k+1}_{vI})$ and $(y^\alpha_{k+1},\pi^{k+1}_\alpha)$ yield a canonically conjugate set of gauge and `graviton' modes. In particular, using that $(Y_{k+1})^e_{vI}=(Y_k)^e_{vI}$, (\ref{32YN}) and $\pi^{k+1}_o=0$, one easily checks that the vertex displacement generators (\ref{14purecons}) are preserved under the 3--2 Pachner moves (\ref{lin32}), yielding
\ba
C^{k+1}_{vI}=(Y_{k+1})^e_{vI}\left(\pi^{k+1}_e-N^{k+1}_{ee'}y^{e'}_{k+1}\right)=\pi^{k+1}_{vI}-N^{k+1}_{vIwJ}y^{wJ}_{k+1}\,.\nn
\ea
Furthermore, noting that $N^{k+1}_{vI\alpha^*}=0$ and $(T_{k+1})^e_{\alpha^*}=0$, the pre--constraint of the 3--2 move (last equation in (\ref{lin32})) trivializes into the momentum conjugate to the `annihilated' or fixed `graviton'
\ba
\pi^{k+1}_{\alpha^*}=0\,.\nn
\ea

For completeness, we mention that thanks to 
\ba
\pi^{k+1}_e=\pi^k_{e}+S^\sigma_{ec}y^c_k=\left(N^k_{ec}+S^\sigma_{ec}\right)y^c_k-\Omega^k_{ea}y^a_0\underset{(\ref{32geom})}{=}N^{k+1}_{ee'}y^{e'}_{k+1}-\Omega^{k+1}_{ea}y^a_0\,,\nn
\ea
the remaining `graviton' momenta can also be written as
\ba
\pi^{k+1}_\alpha=N^{k+1}_{\alpha\beta}y^\beta_{k+1}-\Omega^{k+1}_{\gamma\alpha} y^\gamma_0\,.\nn
\ea


\subsection{The `linearized' 4--1 Pachner move}\label{app_41}

Finally, we back up the claims made in section \ref{sec_41} concerning the linearized 4--1 move. 

Assume that a 4--1 Pachner move can be performed on a hypersurface $\Sigma_k$ and that $(T_k)^c_\Gamma$ has been chosen in conformity with the prescription in section \ref{sec_tmatrix}. The 4--1 move will move an old vertex, which we label by $v^*$ and four old edges adjacent to it, which we index by $o$, into the bulk of the triangulation (see figure \ref{14m}). Accordingly, the index $c$ runs over both $e,o$ and at step $k$ we have
\ba\label{41oldsplit}
y^e_{k}&=&(T_k)^e_{vI}y^{vI}_{k}+(T_k)^e_\alpha y^\alpha_{k}+(T_k)^e_{v^*I}y^{v^*I}_{k}\,,\nn\\
 y^o_{k}&=&(T_k)^o_{vI}y^{vI}_{k}+(T_k)^o_\alpha y^\alpha_{k}+(T_k)^o_{v^*I}y^{v^*I}_{k}\,,\nn\\
y^{vI}_{k}&=&(T_k^{-1})^{vI}_ey^e_{k}+(T_k^{-1})^{vI}_oy^o_{k}\,,\nn\\
 y^\alpha_{k}&=&(T_k^{-1})^\alpha_ey^e_{k}+(T_k^{-1})^\alpha_oy^o_{k}\,,\nn\\
 y^{v^*I}_{k}&=&(T_k^{-1})^{v^*I}_ey^e_{k}+(T_k^{-1})^{v^*I}_oy^o_{k}\,.
\ea
In analogy to the 3--2 move, we must appropriately reduce the old $E\times E$ canonical transformation matrix to a new $(E-4)\times(E-4)$ matrix at $k+1$ which disentangles the surviving gauge and `graviton' variables
\ba\label{41newsplit}
y^e_{k+1}=(T_{k+1})^e_{vI}y^{vI}_{k+1}+(T_{k+1})^e_\alpha y^\alpha_{k+1}\,,\q\,\,\,\, y^{vI}_{k+1}=(T_{k+1}^{-1})^{vI}_ey^e_{k+1}\,,\q\,\,\,\, y^\alpha_{k+1}=(T_{k+1}^{-1})^\alpha_ey^e_{k+1}\,.\nn
\ea
Fortunately, and in contrast to the 3--2 move, for the 4--1 move the reduction of $(T_k)$ to $(T_{k+1})$ turns out to be trivial---just like the linearized equations of motion of this move.

Clearly, at step $k$ we must have $(T_k)^e_{v^*I}=(Y_k)^e_{v^*I}=0$ (components of the $(Y_k)_{vI}$ corresponding to edges not adjacent to the vertex in question vanish). One easily checks that the condition of invertibility of the T--matrix then leads to the following two conditions that must be satisfied:
\ba
\delta^{vI}_{v^*J}=(T_k^{-1})^{vI}_o(T_k)^o_{v^*J}\overset{!}{=}0\,,\q\q \delta^{\alpha}_{v^*I}=(T_k^{-1})^{\alpha}_o(T_k)^o_{v^*I}\overset{!}{=}0\,.\nn
\ea
$(T_k)^o_{v^*I}=(Y_k)^o_{v^*I}$ is a non--degenerate $4\times4$ matrix (there are four linearly independent gauge directions and four edges adjacent to $v^*$). Hence, $(T_k^{-1})^{vI}_o=0$ and $(T_k^{-1})^{\alpha}_o=0$. The conjunction of these results, as one may convince oneself, implies that already the restriction of $(T_k)$ at $k$ to the $(E-4)\times(E-4)$ submatrix 
\ba\label{41Tnewchoice}
(T_{k+1})^e_\Lambda=(T_k)^e_\Lambda\,,\q \q(T_{k+1}^{-1})^\Lambda_e=(T_k^{-1})^\Lambda_e\,,\nn
\ea
where $\Lambda$ only runs over $vI$ and $\alpha$ (but not $v^*I$) yields the desired invertible $(T_{k+1})$ of step $k+1$. The matrix reduction of the 4--1 move is just the time reverse of the matrix extension of the 1--4 move. Specifically, this gives $y^e_{k+1}=y^e_k$, $y^\alpha_{k+1}=y^\alpha_k$ and $y^{vI}_{k+1}=y^{vI}_k$. 

Let us now study the time evolution equations. Firstly, the equations of motion read 
\ba
\left(N^k_{oo'}+S^\sigma_{oo'}\right)y^{o'}_{k}+\left(N^k_{oe}+S^\sigma_{oe}\right)y^e_{k}-\Omega^k_{ao}y^a_0=0\,.\label{41lineom}
\ea
However, for the 4--1 move these are trivial because the results in section \ref{sec_deg} entail that
 \ba\label{n0}
(Y_k)^o_{v^*I}\left(N^k_{oc}+S^\sigma_{oc}\right)=0\nn
\ea
and, as a consequence of $(Y_k)^o_{v^*I}$ being a non--degenerate $4\times4$ matrix, $N^k_{oc}+S^\sigma_{oc}=0$. Similarly, one finds $\Omega^k_{ao}=0$ such that all coefficients in (\ref{41lineom}) vanish. 

Next, we shall examine the consequences of the linearized momentum updating for the gauge and `graviton' variables. We recall that the linearized momentum updating for the 4--1 move is in shape identical to (\ref{lin32}) from the 3--2 move, except that $o$ now labels four edges. We begin by the momenta conjugate to the gauge variables that survive the move. Noting that $\pi^{k+1}_o=0$ and using (\ref{ncantrans}),
\ba
(Y_{k+1})^e_{vI}\pi^{k+1}_e&=&(Y_{k+1})^c_{vI}\pi^{k+1}_c=(Y_{k+1})^c_{vI}\left(\pi^k_c+S^\sigma_{cc'}y^{c'}_{k}\right)\nn\\
&=&\pi^k_{vI}+N^{k}_{vI\alpha}y^{\alpha}_{k}+S^\sigma_{vI\alpha}y^\alpha_{k}+S^\sigma_{vI\tilde{w}J}y^{\tilde{w}J}_{k}\,,\nn
\ea
where $\tilde{w}$ runs over both $v$ and $v^*$. As a consequence of $N^k_{oc}+S^\sigma_{oc}=0$ and despite new internal edges, $N^{k+1}_{ee'}=N^k_{ee'}+S^\sigma_{ee'}$. Since $y^\alpha_{k+1}=y^\alpha_k$, this gives
\ba\label{41momgaup}
\pi^{k+1}_{vI}&:=&(Y_{k+1})^c_{vI}\pi^{k+1}_c-N^{k+1}_{vI\alpha}y^{\alpha}_{k+1}=(Y_{k+1})^e_{vI}\pi^{k+1}_e-(Y_{k+1})^e_{vI}N^{k+1}_{ee'}(T_{k+1})^{e'}_\alpha y^{\alpha}_{k+1}\nn\\
&=&\pi^k_{vI}+S^\sigma_{vI\tilde{w}J}y^{\tilde{w}J}_{k}\,.\nn
\ea
The second equation in the first line shows that, as desired, the new momenta can be computed from the reduced $(T_{k+1})$. Again, solving (\ref{14purecons}) for $\pi^k_{vI}$ converts this apparent evolution equation into the vertex displacement generators of step $k+1$,
\ba\label{41newcon}
C^{k+1}_{vI}=\pi^{k+1}_{vI}-N^{k+1}_{vIwJ}y^{wJ}_{k+1}\,,\nn
\ea
where, by $(Y_{k+1})^e_{v^*I}=0$, $N^{k+1}_{vIv^*J}=0$ and the `annihilated' gauge modes $y^{v^*J}_{k}$ thus drop out. The vertex displacement generators (at neighbouring vertices) are therefore preserved under 4--1 moves as well.

On the other hand, for the momenta conjugate to the four gauge modes associated to $v^*$ one finds (recall that $(Y_{k+1})^e_{v^*I}=0$)
\ba
(Y_{k+1})^c_{v^*I}\pi^{k+1}_c&=&(Y_{k+1})^o_{v^*I}\pi^{k+1}_o=0=(Y_{k})^o_{v^*I}\left(\pi^k_o+S^\sigma_{oc}y^c_k\right)\nn\\
&=&\pi^k_{v^*I}+(Y_{k})^o_{v^*I}\left(\underset{=0}{\underbrace{(N^k_{o\alpha}+S^\sigma_{o\alpha})}}y^\alpha_k+S^\sigma_{o\tilde{v}J}y^{\tilde{v}J}_k\right)\,.\nn
\ea
Using (\ref{ncantrans}), this trivializes the four pre--constraints of the 4--1 move by transforming them into the new momenta conjugate to the `annihilated' gauge modes
\ba
\pi^{k+1}_{v^*I}=0=\pi^k_{v^*I}+S^\sigma_{v^*I\tilde{v}J}y^{\tilde{v}J}_{k}\,.\nn
\ea

Finally, let us address the evolution of the `graviton' momenta. We have
\ba
(T_{k+1})^e_\alpha\pi^{k+1}_e&=&(T_{k+1})^c_\alpha\pi^{k+1}_c=(T_{k})^c_\alpha\left(\pi^k_c+S^\sigma_{cc'}y^{c'}_{k}\right)\nn\\
&=&\pi^k_\alpha+N^{k}_{\alpha \tilde{v}I}y^{\tilde{v}I}_{k}+S^\sigma_{\alpha\beta}y^\beta_{k}+S^\sigma_{\alpha \tilde{v}I}y^{\tilde{v}I}_{k}\,.\nn
\ea
Making use of $N^k_{oc}+S^\sigma_{oc}=0$ and $(Y_{k+1})^e_{v^*I}=0$, one discovers that the `annihilated' gauge modes $y^{v^*I}_{k}$ drop out so that (recall $y^{vI}_{k+1}=y^{vI}_k$)
\ba\label{41momobup}
\pi^{k+1}_\alpha&:=&(T_{k+1})^c_\alpha\pi^{k+1}_c-N^{k+1}_{\alpha vI}y^{vI}_{k+1}=(T_{k+1})^e_\alpha\pi^{k+1}_e-(T_{k+1})^e_\alpha N^{k+1}_{ee'}Y^{e'}_{vI}y^{vI}_{k+1}\nn\\
&=&\pi^k_\alpha+S^\sigma_{\alpha\beta}y^{\beta}_{k}\,.\nn
\ea
The second equality in the first line demonstrates that the new graviton momenta can be determined entirely via $(T_{k+1})$, while the last equality constitutes the usual {\it `graviton' momentum updating}.

\section*{Acknowledgements}
It is a pleasure to thank Bianca Dittrich for discussion. This work has mostly been completed while the author was still at the Institute for Theoretical Physics of Universiteit Utrecht. Research at Perimeter Institute is supported by the Government of Canada through Industry Canada and by the Province of Ontario through the Ministry of Research and Innovation and by the John Templeton Foundation.

\bibliography{bibliography}{}
\bibliographystyle{utphys}

\end{document}